\documentclass[11pt]{amsart}
\usepackage {CJKutf8}
\usepackage{amsmath}
\usepackage{amssymb}
\usepackage{graphicx}
\usepackage[abbrev]{amsrefs} 
\usepackage{hyperref}
\usepackage{bbm}

\usepackage{mathtools} 
\usepackage[scr]{rsfso}
\usepackage[margin=3cm]{geometry}
\usepackage{lmodern}
\usepackage{tikz}
\usetikzlibrary{decorations.pathmorphing, shapes.geometric,arrows, fit}
\usetikzlibrary{positioning}

\usepackage{cancel}

\usepackage{enumitem}

\usepackage{todonotes}

\usepackage{float} 

\newtheorem{theorem}{Theorem}[section]
\newtheorem{corollary}[theorem]{Corollary}
\newtheorem{lemma}[theorem]{Lemma}
\newtheorem{proposition}[theorem]{Proposition}
\theoremstyle{definition}

\newtheorem{remark}[theorem]{Remark}

\numberwithin{equation}{section}

\newcommand{\thmref}[1]{Theorem~\ref{#1}}

\newcommand{\secref}[1]{Section~\ref{#1}}
\newcommand{\lemref}[1]{Lemma~\ref{#1}}
\newcommand{\propref}[1]{Proposition~\ref{#1}}

\newcommand{\figref}[1]{Figure~\ref{#1}}
\newcommand{\corref}[1]{Corollary~\ref{#1}}

\numberwithin{equation}{section}
\numberwithin{theorem}{section}

\DeclareMathOperator{\supp}{supp}

\DeclareMathOperator{\Tr}{Tr}
\newcommand{\1}{\mathbbm{1}}

\newcommand{\be}{\beta}

\newcommand{\del}{\delta}
\newcommand{\eps}{\epsilon}

\newcommand{\lam}{\lambda}

\newcommand{\Lam}{\Lambda}

\newcommand{\cA}{\mathcal{A}}
\newcommand{\cB}{\mathcal{B}}
\newcommand{\cD}{\mathcal{D}}

\newcommand{\cE}{\mathcal{E}}
\newcommand{\cF}{\mathcal{F}}

\newcommand{\cP}{\mathcal{P}}         

\newcommand{\pt}{\partial_t}

\newcommand{\abs}[1]{\ensuremath{\left\lvert#1\right\rvert}}
\newcommand{\norm}[1]{\ensuremath{\left\lVert#1\right\rVert}}
\newcommand{\sbr}[1]{\left[#1\right]}
\newcommand{\Set}[1]{\left\{#1\right\}}
\newcommand{\br}[1]{ \big\langle#1 \big\rangle}
\newcommand{\cbr}[1]{\left(#1\right)}

\newcommand{\od}[2]{\ensuremath{\frac{\md#1}{\md#2}}}

\definecolor{green}{rgb}{0.0, 0.5, 0.5}

\newcommand{\md}{\mathrm{d}}
\newcommand{\me}{\mathrm{e}}
\newcommand{\mi}{\mathrm{i}}

\newcommand{\nc}{\newcommand}
\nc{\ran}{\rangle}
\nc{\lan}{\langle}

\newcommand{\tr}{\mathrm{Tr}}

\usepackage{array}
\usepackage{multirow}

\newcommand{\x}{\lan x\ran}
\newcommand{\y}{\lan y\ran}

\DeclareMathOperator{\dG}{\mathrm{d}\Gamma}

\newcommand{\Norm}[1]{{\left\vert\kern-0.25ex\left\vert\kern-0.25ex\left\vert #1 
		\right\vert\kern-0.25ex\right\vert\kern-0.25ex\right\vert}}

\newcommand{\Nf}[3][]{N_{#2,#3}^{#1}} 
\nc{\cts}[1]{\chi_{ts}(#1)}
\nc{\trzero}[1]{\Tr\sbr{#1\, \rho}}
\nc{\trt}[1]{\Tr\sbr{\tau_t\cbr{#1}\, \rho}}
\nc{\tru}[1]{\Tr\sbr{\tau_u\cbr{#1}\, \rho}}

\newcommand{\tev}[2]{\Tr\sbr{\tau_{#2}\cbr{#1}\,\rho}}

\newcommand{\ip}{\mathrm{int}}

\newcommand{\Ainv}[1]{\cA_{#1}^\mathrm{inv}}

\newcommand{\R}{\mathbb{R}}

\newcommand{\C}{\mathbb{C}}

\begin{document}

	\title[Lieb--Robinson bounds for Bose--Hubbard Hamiltonians]{Lieb--Robinson bounds for Bose--Hubbard Hamiltonians:\\ A review with a simplified proof}

	\author{Marius Lemm}
    \address[M.~Lemm]{Department of Mathematics, University of T\"ubingen, Auf der Morgenstelle 10, 72076 T\"ubingen, Germany}
\email{{\tt marius.lemm@uni-tuebingen.de}}
	\author{Carla Rubiliani}
    \address[C.~Rubiliani]{Department of Mathematics, University of T\"ubingen, Auf der Morgenstelle 10, 72076 T\"ubingen, Germany}
\email{{\tt carla.rubiliani@uni-tuebingen.de}}



	\date{June 30, 2026}

    	\maketitle

		\begin{center}
    \textit{Dedicated to Barry Simon on the occasion of his 80th birthday.}\\
\end{center}
	
	\begin{abstract}	
	We review recent progress on state-dependent Lieb--Robinson bounds for Bose--Hubbard Hamiltonians. In particular, Kuwahara, Vu, and Saito established that, for general bounded-density initial states, the Lieb--Robinson velocity is bounded by $t^{d-1}$  for large times, where $d$ denotes the lattice dimension. We present a  shorter proof of the weaker, but still polynomial velocity bound $t^{d+\epsilon}$.
	\end{abstract}

     \medskip
     
A fundamental question in non-relativistic quantum mechanics is how bounds on the propagation velocity emerge. Indeed, a velocity bound is physically expected because it matches our experience that information tends to propagate at most with a finite, system-dependent speed of sound that is much smaller than the speed of light. The situation is subtle already for quantum dynamics of Schr\"odinger operators $H=-\Delta+V$ on $L^2(\mathbb R^d)$. 
The solution operator $\me^{-\mi tH}$ maps compactly supported initial states to non-compactly supported states at every positive time, thus breaking the PDE notion of finite propagation speed familiar for the wave equation. The resolution to this is well-known: propagation speed bounds should be 
approximate in nature, e.g., by controlling moments of particle position operators instead of wave function supports. At the same time, one cannot expect bounds to be completely independent of the initial state, as can be seen from studying the free case $H=-\Delta$, where the effective propagation speed can be calculated from the initial momentum distribution via Fourier transform.
An early mathematical work on propagation bounds on $L^2(\mathbb R^d)$ was published by Charles Radin and Barry Simon \cite{RS78} in 1978 following earlier work of Hunziker \cite{hunziker1966space}. 
Their argument is simple and characteristically elegant. 
On the Hilbert space $L^2(\mathbb R^d)$, they consider $H=-\Delta+V$ with $V:\mathbb R^d\to\mathbb R$ satisfying the form bound $ |V|\leq -a\Delta+b$ for some $a<1$. 
Denoting $\psi_t=\me^{-\mi tH}\psi$ for initial data $\psi\in H^1(\mathbb R^d)$ and $p=-\mi\nabla$,
the formal version of their argument proceeds as follows. They observe 
\[
\begin{aligned}
\frac{\mathrm{d}}{\mathrm{d}t} \langle \psi_t, x^2 \psi_t\rangle
=\mathrm{i}\langle \psi_t, [H,x^2] \psi_t\rangle
=\mathrm{i}\langle \psi_t, [p^2,x^2] \psi_t\rangle
=2\langle \psi_t, (px+xp) \psi_t\rangle
\leq 4 \|p\psi_t\|_2 \|x\psi_t\|_2,\\
\end{aligned}
\]
which implies that
\[
 \langle \psi_t, x^2 \psi_t\rangle^{1/2}-\langle \psi_0, x^2 \psi_0\rangle^{1/2}\leq 
 \int_0^t 2\|p\psi_s\|_2 \mathrm{d}s.
\]
To estimate the integrand, notice that the form bound implies $-\Delta\leq \tfrac{-\Delta+V+b}{1-a}$. Hence, in the sense of quadratic forms,
\[
2\|p\psi_s\|_2
\leq \frac{2}{1-a} \langle\psi_s, (H+b)\psi_s\rangle
=\frac{2}{1-a} \langle\psi, (H+b)\psi\rangle=:v_\psi
\]
and $v_\psi<\infty$  for $\psi\in H^1(\mathbb R^d)$. To conclude,
\[
\langle \psi_t, x^2 \psi_t\rangle^{1/2}\leq \langle \psi_0, x^2 \psi_0\rangle^{1/2}+v_\psi |t|.
\]

We summarize the argument of Radin-Simon \cite{RS78} as follows: They restrict to physically relevant initial states (in this case, $H^1$, which physically corresponds to bounded kinetic energy) and exploit energy conservation  to control propagation at times $t>0$. 

This is a \textit{ballistic} bound, because it shows that the growth of $\langle \psi_t, x^2 \psi_t\rangle^{1/2}$ (the typical position at time $t$) is bounded by a constant  $v_\psi>0$ times $t$. 
In particular, $v_\psi$ plays the role of a speed bound. The scaling of propagation bounds in $t$ is physically important. For example, a different scaling is expected to occur for random Schr\"odinger operators. 
Namely, they are long conjectured to satisfy a so-called diffusive bound $\langle \psi_t, x^2 \psi_t\rangle^{1/2}\lesssim t^{1/2}$ for all values of the coupling constants multiplying the random potential, a problem that was already emphasized by Barry Simon in his 1984 collection of 15 problems in mathematical physics \cite[Problem 14B]{simon1984fifteen}.
  The investigation of ballistic propagation bounds for Schr\"odinger operators on $L^2(\mathbb R^d)$ has continued to be an influential topic in the following decades. For example, both maximal and minimal velocity estimates played a key role  
in the landmark work of Sigal and Soffer proving  asymptotic completeness \cite{SS87}; see also \cite{G90,HSS99}. Recent works on quantum-mechanical velocity bounds for continuum systems include \cite{arbunich2021maximal,hinrichs2025lieb,sigal2025propagation}. Extensions to nonlinear equations \cite{arbunich2023maximal,liu2025large} and Lindbladian evolutions \cite{breteaux2022maximal,breteaux2024light} were considered as well.

Propagation bounds also play a central role in another area of quantum dynamics, namely for quantum spin systems, which are Hamiltonians describing extensive, locally interacting quantum degrees of freedom on a discrete physical space. For these, Lieb and Robinson \cite{LR72} established a propagation bound which controls commutators of local observables. Indeed, it bounds    $\|[A(t),B]\|$,  where  $A$ and $B$ are two many-body observables that act locally in different regions of space and $A(t)=\me^{\mi tH}A\me^{-\mi tH}$ is the Heisenberg time evolution. The Lieb--Robinson bound serves to control the speed of information propagation in a many-body sense and has turned out to be a decisive tool in quantum information theory proofs. Its first applications in this vein were the proofs of exponential clustering of gapped ground states \cite{hastings2006spectral,nachtergaele2006lieb},  the one-dimensional area law for the entanglement entropy \cite{hastings2007area} and dynamical bounds on generation of entanglement and topological order \cite{bravyi2006lieb}. For   modern accounts of Lieb--Robinson bounds, see \cite{nachtergaele2019quasi,chen2023speed}.

The standard Lieb--Robinson bound controls the operator norm $\|[A(t),B]\|$, which means that it is insensitive to the initial state, a robustness property that has proven useful. The insensitivity of the Lieb--Robinson bound to the initial state can be directly traced back to the fact that the local interactions in quantum spin systems are uniformly bounded. (Indeed, the emergent velocity bound is proportional to the largest operator norm of a local interaction in the system.) However, this insensitivity poses a problem as soon as one aims to prove Lieb--Robinson bounds for systems with \textit{unbounded} local interactions. This is true even for lattice systems with unbounded interactions, which arise naturally for lattice bosons. In \cite{NRSS09}, Nachtergaele, Raz, Sims, and Schlein were able to prove Lieb--Robinson bounds for a class of perturbations of harmonic oscillators. A different, physically  relevant lattice boson model is given by the so-called Bose--Hubbard Hamiltonian on a finite graph $\Lambda\subset \mathbb Z^d$, which can be expressed on the bosonic Fock space by
\[
H=-J\sum_{\substack{x,y\in \Lambda:\\ x\sim y}} a_x^\dagger a_y +U\sum_{x\in \Lambda} n_x(n_x-1)-\mu \sum_{x\in\Lambda} n_x,
\]
where $\{a_x,a_x^\dagger\}_{x\in\Lambda}$ satisfy the usual canonical commutation relations and $n_x=a_x^\dagger a_x$ is the bosonic number operator. (We refer to Section 2 for the precise setup and definitions.) 

For the Bose--Hubbard Hamiltonian, the standard proof techniques for Lieb--Robinson bounds break down. In recent years, the problem of propagation bounds for Bose--Hubbard Hamiltonians and their long-range variants has seen substantial activity and progress \cite{wang2020tightening,yin2022finite,kuwahara2021lieb,FLS2022,LRZS2023,LRZ2023,LRZ2025,Kuwahara2024,lemm2024local,Kuwahara2025}. A related series of works investigated more stringent and robust bounds on macroscopic particle propagation in these models \cite{faupin2022maximal,van2024optimal,li2026macroscopic,faupin2025macroscopic}. For Lieb--Robinson bounds, one overarching insight is remarkably similar to that of Radin-Simon \cite{RS78} discussed above: It is sensible to restrict to a class of physically relevant initial states in which singular behavior is not present. A key challenge is then to prove that this ``good'' behavior of the initial state is sufficiently robust under time evolution. In this paper, we focus on an implementation of this idea in a breakthrough work of Kuwahara-Vu-Saito \cite{Kuwahara2024} whose class of physically relevant states are those of bounded particle density.  In a first step, they prove a bound on the velocity of particle propagation only (as opposed to a Lieb--Robinson bound, which bounds propagation of quantum information  more broadly) and this implies that the density remains well-controlled after finite time $t$. In a second step, the controlled density after time $t$ is used to approximate the dynamics with a truncated dynamics of a spin system of sufficiently large local dimension, to which the usual Lieb--Robinson methodology can then be applied. The work of Kuwahara-Vu-Saito \cite{Kuwahara2024} is rather long and here we provide a simple, self-contained proof of a bosonic Lieb--Robinson bound with a  slightly worse,
but still polynomial velocity scaling ($t^{d+\epsilon}$ versus $t^{d-1}$).

We remark that the scaling is essentially determined by the particle propagation control one uses and how one implements the truncation. 
The quickest, but roughest option is to control particle propagation by a Gronwall argument which leads to exponential-in-time growth of the number of particles. After truncation, this translates into an LR velocity that grows exponentially in time (but is uniform in system size), as was spelled out in \cite{Deuchert2025}. 
The variant we implement here is to prove a bounded speed of particle propagation by the ASTLO (adiabatic space time localization observables) method. 
This leads to local particle numbers being bounded by $\sim t^d$, as at most all particles within a ball of radius $\sim t$ can accumulate at a fixed site within time $t$. 
In \cite{Kuwahara2024}, a further refinement to $t^{d-1}$ is obtained by showing that large LR velocity comes from large particle occupations along a 1D ``information path'' and particle propagation bounds naturally control this accumulation by $t^{d-1}$. 
A proof of $t^d$-scaling is contained in \cite{Kuwahara2024} as Theorem 2 on p.\ 65 of the Supplemental Material. 
On a related note, the assumption on the control of particles in the initial state also plays a role and affects the spatial decay of the obtained Lieb--Robinson bound. 
Here,  we assume existence of a fixed  moment of the particle number operator in the initial state and the spatial decay in the Lieb--Robinson bound is polynomially related to this fixed moment. 
This is another difference to the work \cite{Kuwahara2024}, which assumes existence of exponential moments of the particle number in the initial state and derives exponential spatial decay. 
For us, the polynomial moment assumption is natural because we prove the particle propagation bound by the ASTLO technique (originally developed for long-range interactions \cite{FLS2022,LRZS2023,LRZ2023}). 
We observe here that it allows for a comparatively short proof of the relevant particle propagation bound also in the short-range case. 

All the techniques in this article have appeared before, mainly in \cite{FLS2022,LRZ2023,Kuwahara2024}. Accordingly, we consider this work  a review as far as the mathematical techniques themselves are concerned. Our contribution is to combine them in a slightly different way and streamline them to give what we believe is the currently most direct and simple route to a polynomial light cone for Bose--Hubbard Hamiltonians. 

The paper is self-contained apart from the fact that we cite the standard Lieb--Robinson bound for quantum spin systems \cite{LR72}, specifically the formulation in    \cite{nachtergaele2019quasi}. Another part of the simplification comes from not tracking universal constants and making explicit only the dependence on the density of the initial state. 
 We hope that all of this helps to make the topic more accessible.

	The paper is divided into four main sections.
	\begin{itemize}
		\item In \secref{sec set up}, we introduce the model and setup, together with some useful notation. 
		We conclude the section stating \thmref{teo LRB}, a state-dependent bosonic Lieb--Robinson bound that is the main result of this paper.
		\item In \secref{sec PPB}, we state   the particle propagation bound, \thmref{theo final particle}, a key step for deriving bosonic Lieb--Robinson bounds. Afterwards, we introduce the main ingredient of the proof, the ASTLO. Then we state and prove some basic lemmas that are useful to control the time evolution of such operators. 
        
	\item	In \secref{sec:pp-proof}, we prove \thmref{theo final particle}. In \secref{sec mom 1}, we begin by stating the key  propositions necessary to control the time evolution of the first moment of the number operator.
		They constitute the basis of an induction on moments that allows us to prove bounds on higher moments that we state in \secref{sec PPB higher}.
		Assuming the propositions in \secref{sec mom 1} and \secref{sec PPB higher} are true, we prove \thmref{theo final particle} in \secref{sec proof main part}.
		The proofs of the bounds on the first moment in \secref{sec mom 1} and on the higher moments in \secref{sec PPB higher} can be found in \secref{sec rpoof first mom} and \secref{sec proof high PPB}, respectively.
		
		\item In \secref{sec derive LRB}, we prove \thmref{teo LRB}, making use of the particle propagation bounds obtained in \secref{sec PPB}. We start the section by introducing the truncated dynamics and compare them to the full dynamics in \propref{prop H to bar H}, which is proved in \secref{sec Proof of prop H to bar H}. For the truncated dynamics, one can locally connect back to a   standard Lieb--Robinson bounds; this is stated as \propref{prop bar H to H R} and proved in \secref{sec Proof of prop bar H to H R}. From this, we conclude \thmref{teo LRB}.

	\end{itemize}
\section{Setup and main result}\label{sec set up}
\subsection{Setup and notation}

Consider a finite subset $\Lam\subset \mathbb{Z}^d$ equipped with euclidean metric $\abs{\cdot}$. Bosonic Fock space is then defined as 
\begin{align}
    \mathcal F=  \C\oplus\bigoplus_{N=1}^\infty \ell^2_s(\Lam^N),\notag
\end{align}
where $\ell^2_s(\Lam^N)$ is the Hilbert space of permutation-symmetric squared summable sequences on $\Lam^N$.
For every $x,y \in\Lam$, we define the following operators acting on $\cF$
\begin{align}\label{T, V def}
	T_{xy}:=Ja_x^\dagger a_y,\qquad V_{xy}:= U n^{p/2}_xn^{p/2}_y-\mu( n_x+n_y),
\end{align}
where $J,U,\mu\in\mathbb{R}$ and $p\ge 1$. 
In \eqref{T, V def},  $a^\dagger_x$ and $a_x$ are the bosonic creation and annihilation operators, thus they satisfy the canonical commutation relations, $\sbr{a_x,a_y}=0=\sbr{a^\dagger_x,a^\dagger_y}$ and $\sbr{a_x,a_y^\dagger}=\del_{x,y}$. 
They define the local number operator as $n_x:=a_x^\dagger a_x$ and $N_X:=\sum_{x\in X}n_x$, for $X\subset \Lam$. For simplicity, we write $N_\Lam\equiv N$.
For a given subset $X\subset \Lambda$, we define
\begin{align}
	T_X:=\sum_{\substack{x,y\in X \\x\sim y}}T_{xy},\quad V_X:=\sum_{\substack{x,y\in X \\x\sim y}}V_{xy},\label{T V on X}
\end{align}
where $x\sim y$ indicates $\abs{x-y}=1$,
and we consider the following nearest-neighbors Bose--Hubbard-type Hamiltonian
\begin{align}\label{ham}
	H_X:=T_X+V_X.
\end{align}
 For ease of notation, we write $H\equiv H_\Lam$, $T\equiv T_\Lam$, and $V\equiv V_\Lam$ .
Given $X\subset \Lam$ and $R>0$, we denote its $R$-enlargement by 
\begin{align}
	X[R]:=\Set{x\in\Lam\,:\,d\cbr{x,X}\le R},\notag
\end{align}
where $d(x,X):=\inf\Set{\abs{x-y}\;:\;y\in X}$.	
Furthermore we indicate its diameter by
\begin{align}
    d(X):=1+\max\Set{\abs{x-y}\,:\, x,y\in X}.
\end{align}
For ease of notation, when $X=\Set{x}$ for some $x\in\Lam$, we write $\Set{x}[R]\equiv B_R(x)$. In particular, when $x$ is the origin we write $B_R(0)\equiv B_R$.

We denote by $\cB$ the algebra of bounded operators on Fock space, by $\cA_X\subset \cB$ the algebra of bounded operators that are supported on $X\subset\Lam$, and by $\Ainv{X}\subset\cA_X$ the algebra of operators $A\in\cA_X$ such that they commute with the total number operator. 
Notice that, for $A\in\Ainv{X}$, it holds
\begin{align}
	\sbr{A,N_X}=\sbr{A,N}-\sbr{A,N_{X^c}}=0.\notag
\end{align}
Given an operator $A\in\cB$, we denote by $D(A)$ its domain, by $\norm{A}$ its operator norm, and if $A$ is trace class, by $\norm{A}_1:=\Tr\sbr{\sqrt{A^\dagger A}}$ its trace norm. 
We denote by $\tau$ the time evolution induced by the Hamiltonian H, and, for an operator $A$ acting on Fock space, we write
 \begin{align}\label{tau H}
	\tau_t(A)=\me^{\mi t H}A\me^{-\mi t H},\qquad t\in\R.
 \end{align}
If $A\in\Ainv{X}$, with $X\subset \Lam$ compact, we can compare \eqref{tau H} with $\tau_t^R(A)$, where $R>0$ and 
\begin{align}
  \tau^R(A):= \me^{\mi t H_{X[R]}}A\me^{-\mi t H_{X[R]}}.\label{tau R}
\end{align}
In \eqref{tau R}, the operator $A$ is time evolved by the Hamiltonian $H$ restricted on the region $X[R]$, thus $\tau^R(A)$ is also a local operator supported on $X[R]$.

 We focus on initial states in the space,
 \begin{align}
     \cD_\eta:=\Set{\rho \,:\, \rho^\dagger=\rho\,,\,\rho \text{ is trace class on }\cF,\, \Tr[\rho]=1,\,\Tr\sbr{H\rho}<\infty,\,\Tr\sbr{N^\eta\rho}<\infty },\notag
 \end{align}
 for some integer $\eta>0$.
Additionally, we say $\rho$ satisfies the \textit{controlled density assumption} for some $\eta\ge 1$ and $\lam>0$, if 
 \begin{align}\label{CD rho}
	\Tr[N_{B_r(x)}^{\zeta}\rho]\le (\lam r^d)^{\zeta}\qquad \forall \;r\ge 1,\, \,x\in\Lam, \,{\zeta}\le \eta.
 \end{align}

\subsection{Main result}
We state the main result of this paper, a Lieb--Robinson bound that allows to approximate the time evolution of compactly supported operators by local operators.
Namely, we consider an operator $A\in\Ainv{X}$ with $X\subset \Lam$ compact, and we aim to approximate its time evolution through the full Hamiltonian, $H$, by its evolution through the Hamiltonian $H_{X[R]}$, for  $R>2$.
\begin{theorem}[Lieb--Robinson bound]\label{teo LRB}
    Consider $\eta\ge 2(d+1)$.
	There exists  a positive $C=C(d, \eta,  J)$, such that, for every initial state $\rho\in\cD_\eta$ satisfying \eqref{CD rho} for $\eta$ and some $\lam >0$,   the following holds
\begin{align}\label{main teo eq}
		\norm{\cbr{\tau_t(A)-\tau^R_t(A)}\rho}_1\le C\norm{A}\abs{X}\cbr{ \cbr{\lambda \,d(X)^{d} \,R^{-1}\cbr{(\abs{t}+1)R^{2/\eta}}^{d+1}}^{\eta/2}+R^d\me^{- R/2}},
	\end{align}
    for all $R>2$, $t\in\R$, $X\subset\Lam$ compact, and $A\in\Ainv{X}$. 
\end{theorem}
This shows that the norm on the l.h.s. of \eqref{main teo eq} is polynomially suppressed when $R>\abs{t}^{d+1+\epsilon}$ with $\epsilon:=(d+1)^2/(\eta/2-(d+1))$ small for  $\eta$ large enough. 
We conclude the section with some remarks on the main theorem.
\begin{remark}\begin{enumerate}[label=(\roman*)]
\item As we remarked above, the time scaling of the velocity, $v\sim t^{d+\epsilon}$, and the polynomial decay outside the light-cone that we have obtained in \thmref{teo LRB} are weaker than the bound $v\sim t^{d-1}$ obtained in \cite{Kuwahara2024}.
This is illustrated in \figref{fig:placeholder}.
\item 
    The constants above are independent of system size and of the potential parameters, $U$ and $\mu$.
\item 
    The dependence on the density of particles $\lambda$ in the initial state is explicit. Naturally, the bound improves for small $\lambda$. 

\end{enumerate}
\end{remark}

\begin{figure}
    \centering

\begin{tikzpicture}[scale=1]

\def\Rmax{5}
\def\Tmax{5}

\draw[->, thick] (-\Rmax,0) -- (\Rmax,0) node[right] {$\Lambda$};

\fill[green,opacity=.2]
(-\Rmax,\Tmax)
--
plot[domain=(-\Rmax^(1/3)):0, smooth, variable=\y]
	({\y^(3)},-\y)	
--
plot[domain=0:(\Rmax^(1/3)), smooth, variable=\y]
	({\y^(3)},\y)
	
--
(\Rmax,\Tmax)
-- cycle;

\fill[blue,opacity=.2]
	(-\Rmax,\Tmax)
	--
	plot[domain=(-\Rmax^(1/1.5)):0, smooth, variable=\y]
        ({\y^(1.5)},-\y)	
	--
    plot[domain=0:(\Rmax^(1/1.5)), smooth, variable=\y]
        ({\y^(1.5)},\y)
		
    --
    (\Rmax,\Tmax)
    -- cycle;

\draw[green!60, thick, variable=\y]
plot[domain=(-\Rmax^(1/3)):0, smooth, variable=\y]
	({\y^(3)},-\y)	
--
plot[domain=0:(\Rmax^(1/3)), smooth, variable=\y]
	({\y^(3)},\y);

\draw[blue!60, thick, variable=\y]

plot[domain=(-\Rmax^(1/1.5)):0, smooth, variable=\y]
	({\y^(1.5)},-\y)	
--
plot[domain=0:(\Rmax^(1/1.5)), smooth, variable=\y]
	({\y^(1.5)},\y);

 \draw[->, thick] (0,0) -- (0,\Tmax+0.3) node[above] {$t$};

\foreach \x in {-7,...,7}
{
    \fill[gray] (2*\x/3,0) circle (2pt);

}

\node[blue!70] at (4,3) {$R\sim t^d$};
\node[green!70] at (4,1.9) {$R\sim t^{d+1+\eps}$};
 \end{tikzpicture}

     \caption{Comparison between the light cones obtained in \thmref{teo LRB}, that scales as $R\sim t^{d+1+\eps} $, and  \cite{Kuwahara2024}, with scaling $R\sim t^d$. }
    \label{fig:placeholder}
\end{figure}
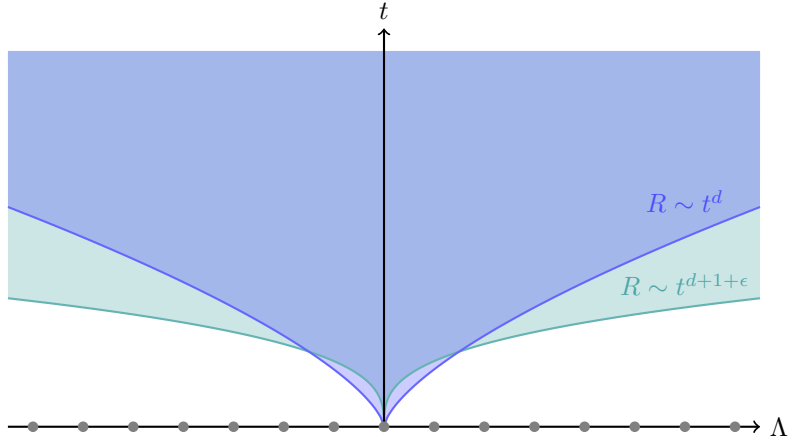

\section{Particle propagation bounds}\label{sec PPB}
In this section, we control the time evolution of the moments of the number operator restricted on certain regions of the lattice. Afterwards, we define our main proof tool, the ASTLOs, and review their basic properties.

\subsection{The particle propagation bound}
The following theorem is the main result of this section
\begin{theorem}[Particle propagation bound]\label{theo final particle}
     Fix $v>2d\abs{J} $, $\delta_0\in(0,1)$, $\beta\ge1$, and $\eta\ge1$.
     There exists a positive $C=C(d,J,v,\delta_0,\beta,\eta)$ such that, for all $\rho\in \cD_\eta$ satisfying \eqref{CD rho} for $\eta $ and some $\lambda>0$, the following holds
	\begin{align}\label{eq final particles}
		\tev{N^\eta _{B_r(x)}}{t}\le C\trzero{N^\eta _{B_R(x)}}+C\lam^\eta(R-r)^{-\beta+d\eta},
	\end{align}
	for all $x\in\Lambda$, $R>r\ge0$ with $R-r>\max\Set{1,\delta_0r}$, and $ v \abs{t}\le (R-r)$. 
\end{theorem}

Notice that \eqref{eq final particles}, together with \eqref{CD rho} and $R-r>1$, yields
\begin{align}
	\tev{N^\eta_{B_r(x)}}{t}\le C_1\lam^\eta R^{d\eta},\notag
\end{align}
for some constant $C_1$ independent of system size, $r$, and $R$. 
Furthermore, for $r=0$ fixing $R=vt$, we obtain
\begin{align}
	\tev{ n_x^\eta}{t}\le C_2\lam^\eta t^{d\eta}.\notag
\end{align}
\subsection{ASTLO}
The main tool in the proof of \thmref{theo final particle} is the so-called ASTLO (\textit{adiabatic space-time localization observable}) method that we review in this section.
Write 
\begin{align}
	\kappa:=2d\abs{J},\notag
\end{align}
then, for every $v>\kappa$, we define
\begin{align}
	&\tilde v:= \frac{v+\kappa }{2}\in(\kappa ,v),\notag\\
	&\eps:=v-\tilde v.\label{eps def}
\end{align}
Consider then the following family of smooth cut-off functions
\begin{align}
	\cE\equiv \cE_\epsilon:=\left\{
	  f\in C^{\infty} (\mathbb{R})\left\vert \begin{aligned} &f\geq0,\: f\equiv 0 \text{ on } (-\infty,\epsilon/2],f\equiv 1 \text{ on } [\epsilon,\infty) \\ &f'\geq 0,\:\sqrt{f'}\in C_c^{\infty} (\mathbb{R}),\:\mathrm{supp} f'\subset(\epsilon/2,\epsilon)
	  \end{aligned}\right.\right\}.  \notag
  \end{align}
  Notice that 
  \begin{align}\label{prop C}
	\text{for every }f_1,f_2\in\cE \text{ there exist }C>0,\,f_3\in\cE \text{ such that  } f_1+f_2\le C f_3.
\end{align}
Consider $R>r\ge0$ and define 
\begin{align}
	s:=(R-r)/v.\label{s def}
\end{align}
For every $t\ge 0$, we can rescale a function $\chi\in\cE$ as follows
\begin{align}
    \chi_{ts}(x)=\chi\left(\frac{R-\tilde vt-|x|}{s}\right),\notag\qquad x\in\mathbb R^d.
\end{align}
We define the ASTLO by second quantization of these rescaled functions, namely
\begin{align}\label{ASTLO}
	\Nf[]{\chi}{ts}:=\md\Gamma(\chi_{ts}),
\end{align}
where for every operator $A$ acting on the one particle Hilbert space $\ell^2(\Lambda)$ we define the second quantization map as follows
\begin{align}\label{dGdef}
	\md\Gamma(A):=\sum_{x,y\in\Lam}A_{xy} a_x^\dagger a_y.
\end{align}
Notice that, for a function $f\in \ell^\infty(\Lam)$, the relation above reduces to
\begin{align}
    \md\Gamma(f)=\sum_{x\in\Lam}f(x) n_x.\notag
\end{align}

\subsection{Basic lemmas}
We introduce some basic lemmas whose proofs can be found in previous works \cite{FLS2022,LRZ2023,LRZ2025}, and that we include here for completeness.

\begin{lemma}[Geometric properties \cite{LRZ2025}*{Lemma~3.1}]\label{prop geom prop}
	Given $R>r>0$ and function $f\in\cE$, the following holds for all $t\le s$:
	\begin{align}
		\Nf[]{f}{0s}\le&{N_{B_{R}}} \label{f and B at t=0},\\
		\Nf[]{f}{ts}\,\ge&{N_{B_{r}}}\label{f and b at t}.
	\end{align}
\end{lemma}

\begin{lemma}[Symmetrized Taylor expansion, \cite{FLS2022}*{Lem.~2.2}]\label{lemm sym taylor}

		Let $\beta\ge1$ be an integer and $\chi \in\cE$. Define $u:=(\chi ')^{\frac{1}{2}}$. Then there exist a constant $C_{\chi,\beta}$ and, for $\beta\ge2$, functions $j_k\in\cE$ and constants $C_{\chi ,k}$, $2\leq k\leq \beta$, such that,  for  all $x,y\in\mathbb{R}$, 
		\begin{equation}
			\chi (x)-\chi (y)=(x-y)u(x)u(y)+\sum_{k=2}^\beta (x-y)^k h_k (x,y)+(x-y)^{\beta+1}R_\beta(x,y),\notag
		\end{equation}
		where the sum should be dropped for $\beta=1$, and
		\begin{align}
			|h_k(x,y)|\leq& C_{\chi ,k} \Tilde{u}_k(x)\Tilde{u}_k(y),\qquad 2\leq k\leq \beta,\notag\\
			\abs{R_{\beta}(x,y)}\le& C_{\chi ,\beta},\notag
		\end{align}
with 
\begin{align}
    \tilde u_k:=(j_k')^{\frac{1}{2}}, \qquad 2\le k\le \beta.\notag
\end{align}
\end{lemma}

\begin{lemma}[Commutator expansion, \cite{FLS2022}*{Lem.~A.2}]\label{lemma comm exp}
	Let $H$ be as in \eqref{ham}. Let $g\in\ell^\infty(\Lam)$. In the sense of forms on $D(H)\cap D(N)$, we have
	
	\begin{equation}\label{HRf-com}
		\sbr{H,\dG(g)}=-J\sum_{x,y\in\Lambda}\cbr{g(x)-g(y)}a_x^{*}a_y.
	\end{equation}
\end{lemma}

The proof of this lemma can be summarized by $\sbr{H,\dG(g)}=\sbr{\dG(JA),\dG(g)}=J\dG(\sbr{A,g})$ with $A$ the adjacency matrix of the graph. We refer to \cite{FLS2022}*{Lem.~A.2} for the details.

\subsection{Proof of \lemref{prop geom prop}}
The proof of \lemref{prop geom prop} can be found in \cite{LRZ2023,LRZ2025}, but we include it in this section for the convenience of the reader.
\begin{proof}
We start by noticing that for a function $\chi\in\cE$ and $x\in\mathbb{R}$ it holds
\begin{align}\label{control function}
	\1_{x\ge \epsilon}\le \chi(x)\le \1_{x\ge \epsilon/2}.
\end{align}
The first inequality in \eqref{control function}, together with the definition of $\epsilon$ \eqref{eps def}, $s$ \eqref{s def}, and $s\ge t$, yields
\begin{align}
	\trt{N_{\chi,ts}}\ge \sum_{\abs{x}\le R-\tilde vt-\eps s}\trt{n_x}\ge \sum_{\abs{x}\le r}\trt{n_x}=\tev{N_{B_r}}{t}.\notag
\end{align}
Assuming $t=0$, the second inequality in \eqref{control function} leads to
\begin{align}
	\trzero{N_{\chi,0s}}\le \sum_{\abs{x}\le R-\eps s/2}\trzero{n_x}\le \sum_{\abs{x}\le R}\trzero{n_x}=\trzero{N_{B_R}}.\notag
\end{align}
\end{proof}

\subsection{Proof of \lemref{lemm sym taylor}}
The proof of \lemref{lemm sym taylor} can be found in \cite{FLS2022} and we write it here for completeness. 

\begin{proof}
	We start by Taylor expanding the difference we aim to bound
	\begin{align}\label{after 1 taylor}
		\chi(x)-\chi(y)=\sum_{k=1}^\beta \frac{(x-y)^k}{k!}\chi^{(k)}(x)+C_{\beta,\chi}(x-y)^{\beta+1}.
	\end{align}
We now aim to symmetrize the expansion above with respect to $x$ and $y$. To this end we write
\begin{align}
	\chi'(x)=u^2(x)=u^2(x)\pm u(x)u(y)=u(x)u(y)+u(x)\cbr{u(x)-u(y)}.\notag
\end{align}
Thus, \lemref{lemm sym taylor} holds for $\beta=1$. 
To obtain the desired result for $\beta\ge 2$, we further Taylor expand the difference $u(x)-u(y)$ until order $\beta-1$, and \eqref{after 1 taylor} becomes
\begin{align}\label{after taylor 2}
	\chi(x)-\chi(y)=(x-y)u(x)u(y)+\sum_{k=2}^\beta(x-y)^k\cbr{ \frac{\chi^{(k)}(x)}{k!}+\frac{u(x)u^{(k-1)}(x)}{(k-1)!}}+C_1(x-y)^{\beta+1}	.
\end{align}
Now we iterate the above procedure to bound the higher order terms. 
Notice that if a function $f$ is such that $f\in C^\infty(\mathbb{R})$, $f\ge 0$, $\sqrt{f}\in C^\infty(\mathbb{R})$ with $\supp f\subset (\epsilon/2,\epsilon)$, then there exist $g\in \cE$ and constant $C$ such that $f\le Cg'$. 
In fact, 
\begin{align}
	\cbr{\int_{-\infty}^{\infty}f(s)\md s}^{-1}\cbr{\int_{-\infty}^{x}f(s)\md s}\in\cE.\notag
\end{align}
We call functions as $f$ above \textit{admissible}.
Furthermore, given $f,g\in C_0^\infty(\mathbb{R})$ we write $f \prec g$ if $g\equiv 1$ on $\supp f$.
With this at hand, we can start analyzing the $k=2$ term in \eqref{after taylor 2}. 
Denote by 
\begin{align}
	w_2:=\frac{\chi^{(2)}(x)}{2}+u(x)u'(x),\notag
\end{align}
and for some $v_2\in C_0^\infty(\mathbb{R})$ with $\supp v_2\subset (0,\epsilon)$ and $\chi'\prec v_2$ we write
\begin{align}\label{after taylor 3}
	(x-y)^2\cbr{ \frac{\chi^{(2)}(x)}{2}+u(x)u'(x)}&= (x-y)^2v_2(x)w_2(x)v_2(y) \notag \\
													&\quad + (x-y)^2v_2(x)w_2(x)(v_2(x)-v_2(y)) \notag\\
												   &= (x-y)^2v_2(x)w_2(x)v_2(y)\\
													&\quad + v_2(x)w_2(x)\sum_{k=3}^{\beta}\frac{(x-y)^{k}}{(k-2)!}v_2^{(k-2)}(x)+C_2(x-y)^{\beta+1}\notag,
\end{align}
where we Taylor expanded the difference above until order $\beta-2$.
Inequality \eqref{after taylor 2}, together with \eqref{after taylor 3}, yields
\begin{align}
	\chi(x)-\chi(y)&=(x-y)u(x)u(y)+(x-y)^2v_2(x)w_2(x)v_2(y)\notag\\
	&\quad +\sum_{k=3}^\beta(x-y)^k\cbr{ \frac{\chi^{(k)}(x)}{k!}+\frac{u(x)u^{(k-1)}(x)}{(k-1)!}+v_2(x)w_2(x)\frac{v_2^{(k-2)}(x)}{(k-2)!} }\notag	\\
	&\quad +C_3(x-y)^{\beta+1}.\notag
\end{align}
We iterate the above procedure to control the terms for $k=3,\dots, \beta$, and we obtain
\begin{align}
	\chi(x)-\chi(y)&=(x-y)u(x)u(y)+\sum_{k=2}^{\beta}(x-y)^kv_k(x)w_k(x)v_k(y)  +C(x-y)^{\beta+1},\notag
\end{align}
for some $v_k$ smooth and non-negative such that $\supp v_k\subset (0,\epsilon)$ and $f'\prec v_k$, and $w_k$ smooth with $\supp w_k\subset(0,\epsilon)$.
Notice that, since $v_k$ are admissible functions, 
\begin{align}
	\abs{v_k(x)w_k(x)v_k(y) }\le C_k\, v_k(x)v_k(y)\le C'_k \, ({j_k}'(x))^{1/2}  ( {j_k}'(y))^{1/2},  \notag
\end{align}
for some $ j_k\in\cE$. This concludes the proof.
\end{proof}

\section{Proof of particle propagation bound}
\label{sec:pp-proof}
\subsection{First moment of the number operator}\label{sec mom 1}
In this section, we control the time evolution of the first moment of the number operator, this will constitute the induction basis needed to prove a similar bound for higher moments.
We start by controlling the time derivative of the ASTLO in the following proposition.

\begin{proposition}[Recursive structure]\label{prop recurs stru}
	Consider $v>\kappa $, $\chi\in\cE$. Then, for $\beta\ge 1$, there exists $C=C(d,\chi,\beta,J)>0$, and, for $\beta\ge 2$, a function $\tilde \chi\in\cE$ such that, for all states $\rho\in\cD_1$, and for all $t,r,R>0$ with $R>r$ and $ R-r\ge1$, it holds
\begin{align}\label{eq recurs stru}
	\od{}{t}\tev{\Nf[]{\chi}{ts}}{t}\le  \frac{\kappa -\tilde v}{s}\tev{\Nf[]{\chi'}{ts}}{t}+\frac{C}{s^{2}}\tev{\Nf[]{\tilde \chi'}{ts}}{t}+\frac{C}{s^{\beta+1}}\tev{N_{B_{R+1}}}{t},
\end{align}
where for $\beta=1$, the second summand on the r.h.s. of \eqref{eq recurs stru} should be dropped. Recall $s:=(R-r)/v$ and $\tilde v:=(v-\kappa)/2$. 
\end{proposition}

Bootstrapping the integrated version of \eqref{eq recurs stru} and applying \lemref{prop geom prop} yields  the following corollary. 
\begin{corollary}[Bootstrapping]\label{cor boot}
	Consider $\beta\ge 1$ and  $v> \kappa $. Then, there exists $C=C(d,v,\beta,J)>0$ such that, for all states $\rho\in\cD_1$, and for all $R>r>0$ such that $ R-r\ge 1$, it holds 
	\begin{align}\label{eq boot}
		\tev{N_{B_r}}{t}\le \cbr{1+\frac{C}{s}} \trzero{N_{B_R}}+C\,s^{-\beta}\sup_{0\le u\le t}\tev{N_{B_{R+1}}}{u},
	\end{align}
	for all $0\le t\le s$.
\end{corollary}

By a downwards multi-scale induction we can control the remainder term appearing in \eqref{eq boot} and we prove the following proposition.

\begin{proposition}[Remainder bound]\label{prop remainder}
    Fix $v>\kappa $, $\delta_0\in(0,1)$, and $\beta>1$. 
    There exists a positive $C=C(v,\delta_0,\beta,J,d)$ such that for all $\rho\in\cD_1$, satisfying \eqref{CD rho} for some $\lam>0$ and $\eta=1$,  all $R>r\ge0$ with $R-r>\max\Set{1,\delta_0r}$, the following holds
\begin{align}\label{eq remainder control}
	\tev{N_{B_{r}}}{t}\le (1+s^{-1}C)\trzero{N_{B_{R}}}+ C\lambda s^{-\beta+d}
\end{align}
for all $0\le t\le s$.
\end{proposition}

\subsection{Higher moments}\label{sec PPB higher}
Using the result in the previous section as the base case, we build an induction on the moment parameter $\eta$ to show the following proposition.

\begin{proposition}[Recursive structure for $\eta>1$]\label{prop recurs stru p}
	Consider $v>\kappa $ and $\chi\in\cE$. For $\beta\ge 1$, there exists $C=C(d,\chi,\beta,J)>0$, and, for $\beta\ge 2$, a function $\tilde \chi\in\cE$ such that, for every  state $\rho\in\cD_\eta $ for some $\eta>1$, and for all $t,r,R>0$ with $R>r$ and $R-r\ge 0$,  the following holds
\begin{align}\label{eq recurs stru p}
	\od{}{t}\tev{\Nf[\eta]{\chi}{ts}}{t}\le & \;\eta(\eta-1)\tilde v  s^{-1}\tev{N_{\chi',ts}(N_{\chi,ts}+1)^{\eta-2}}{t}\\
	&+\eta(\kappa -\tilde v) s^{-1}\tev{N_{\chi',ts}(N_{\chi,ts}+1)^{\eta-1} }{t}
	\notag\\
	&+\eta \,C\;s^{-2}\tev{N_{\tilde \chi',ts} (N_{\chi,ts}+1)^{\eta-1}}{t}\notag\\
	&+\frac{ \eta \, C}{s^{\beta+1}} \tev{N_{B_{R+1}}(N_{B_{R+1}}+1)^{\eta-1}}{t}\notag.
\end{align}
\end{proposition}

Using a bootstrapping strategy as the one employed in the proof of \corref{cor boot}, coupled with an induction on the moment parameter $\eta$, we can prove the following corollary.

\begin{corollary}[Bootstrapping for $\eta>1$]\label{cor boot p}
	Consider $\beta\ge 1$, $v> \kappa $, $\eta>1$. There exists ${C=C(d,v,\beta,J,\eta)>0}$ such that, for every $\rho\in\cD_\eta $ and for all $R>r\ge0$ such that $R-r\ge 0$, it holds 
	\begin{align}\label{eq boot p}
		\tev{N^\eta _{B_r}}{t}\le C\cbr{1+s^{-1}}\trzero{N^\eta _{B_R}}+ \frac{ C}{s^{\beta}}\sup_{0\le u\le t}\tev{N^{\eta}_{B_{R+1}}}{u},
	\end{align}
	for all $0\le t\le s$.
\end{corollary}

To control the third term on the r.h.s. of \eqref{eq boot p} we set up a downwards induction on scales as the one in the proof of \propref{prop remainder}.

\begin{proposition}[Remainder Bound for $\eta>1$] \label{prop rem p}
    Fix $v>\kappa $, $\delta_0\in(0,1)$, $\beta\ge1$, and $\eta>1$.
    There exists $C=C(v, \delta_0, \beta, J,d,\eta)>0$ such that, for all $\rho\in\cD_\eta $ satisfying \eqref{CD rho} for $\eta$ and  some $\lam>0$,  the following holds
	\begin{align}\label{eq contr rem}
		\tev{N^\eta _{B_r}}{t}\le C\trzero{N^\eta _{B_R}}+C s^{-\beta+d\eta}\lam^\eta,
	\end{align}
	for all $R>r\ge0$ with $R-r>\max\Set{1,\delta_0r}$ and $0\le  t\le s$. 
\end{proposition}

\subsection{Proof of \thmref{theo final particle}}\label{sec proof main part}

\begin{proof}
	Notice that by translating the ASTLO and following the same steps in the proofs of \propref{prop recurs stru} --\propref{prop rem p}, we obtain, under the assumptions of \thmref{theo final particle},
		\begin{align}\label{final part pos t with still R-r}
			\tev{N^\eta _{B_r(x)}}{t}\le C'\cbr{1+\cbr{R-r}^{-1}}\trzero{N^\eta _{B_R(x)}}+C\lam^\eta(R-r)^{-\beta+d\eta},
		\end{align}
    for all $x\in\Lambda$ and  $t\ge 0$ such that $vt\le R-r$. 
    To obtain \thmref{theo final particle} for $t\le 0$, we write
    \begin{align}\label{final part neg t with still R-r}
       \tev{N^\eta _{B_r(x)}}{t}=\Tr\sbr{\me^{\mi tH}N^\eta _{B_r(x)}\me^{-\mi tH} \rho } = \Tr\sbr{\me^{\mi (-t)(-H)}N^\eta _{B_r(x)}\me^{-\mi (-t)(-H)} \rho }.
    \end{align}
    Since there were no assumptions on the sign of $H$, we can apply \eqref{final part pos t with still R-r} to the r.h.s. of \eqref{final part neg t with still R-r}. Thus, \eqref{final part pos t with still R-r} holds also for negative times.
	Since $R-r\ge 1$, \eqref{eq final particles} holds.
	\end{proof}

\subsection{Proof of the first moment bound}\label{sec rpoof first mom}
In this section we prove \propref{prop recurs stru}, \corref{cor boot}, and \propref{prop remainder}. 
The proofs follow the reasoning of \cite{FLS2022,LRZS2023}.
In what follows the constant are allowed to change from line to line, but they remain independent of system size throughout.

\subsubsection{Proof of \propref{prop recurs stru}}

\begin{proof}
	We start by computing the Heisenberg derivative of the ASTLO,
	\begin{align}\label{Heis der}
		D\Nf[]{\chi}{ts}:=\pt\Nf[]{\chi}{ts}+\mi\sbr{H,\Nf[]{\chi}{ts}}.
	\end{align}
The first term in \eqref{Heis der} can be easily computed,
\begin{align}\label{first term derivative}
	\pt\Nf[]{\chi}{ts}=-\frac{\tilde v}{s}\,\Nf[]{\chi'}{ts}.
\end{align}
Thanks to the structure of the ASTLO \eqref{ASTLO}, \lemref{lemma comm exp} yields
\begin{align}\label{deriv after comm exp}
	\mi\sbr{H,\Nf[]{\chi}{ts}}=\mi J\sum_{x\sim y}\cbr{\cts x-\cts y}a_x^\dagger a_y.
\end{align}
Notice that $\cts x-\cts y\ne 0$ implies either $x$ or $y$ lies in $\supp \chi_{ts}\subset B_R$. This fact, together with \lemref{lemm sym taylor}, yields the following symmetrized expansion
\begin{align}\label{symm exp}
	|\cts x -\cts y|
	\le& |\cts x -\cts y| \cbr{\1_{|x|\leq R }+\1_{|y|\leq R}}\notag\\
	\le& 	\frac{|x-y|}{s} u_{ts}(x)u_{ts}(y)\notag\\
	&+\sum_{k=2}^\beta C_{\chi,k}\cbr{\frac{|x-y|}{s}}^k 
	u_{k,ts}(x)u_{k,ts}(y)\notag\\
	&+C_{\chi,\beta+1}\cbr{\frac{|x-y|}{s}}^{\beta+1}  \cbr{\1_{|x|\leq R }+\1_{|y|\leq R}},
\end{align}
where the sum should be dropped for $\beta=1$, $\mathbbm 1_X$ is the characteristic function of $X\subset\Lambda$, and, for $2\le k\le \beta$,
\begin{align}
	u_{ts}=(\sqrt{\chi'})_{ts},\qquad u_{k,ts}\equiv (u_k)_{ts}=(\sqrt{j_k'})_{ts},	\notag
\end{align}
with $j_k\in\cE$.
Fix a state  $\rho\in\cD_1$. 
Lines \eqref{deriv after comm exp} and \eqref{symm exp} together yield 
\begin{align}
	\trzero{\mi\sbr{H,\Nf[]{\chi}{ts}}}\le& 
	\abs{J}\sum_{x\sim y}\abs{\cts x-\cts y}\abs{\trzero{a_x^\dagger a_y}}\notag\\
	\le&\,\mathrm{I}+\mathrm{II}+\mathrm{Rem},\notag
\end{align}
where
\begin{align}
   & \mathrm{I}:=	s^{-1}\sum_{x\sim y} u_{ts}(x)u_{ts}(y)\abs{\trzero{a_x^\dagger a_y}}\label{I def}\\
   & \mathrm{II}:=\sum_{k=2}^\beta \frac{C'_k}{s^{k}} \sum_{x\sim y}u_{k,ts}(x) u_{k,ts}(y) \abs{\trzero{a_x^\dagger a_y}}\label{II def}\\
   & \mathrm{Rem}:=\frac{C'_{\beta+1}}{s^{\beta+1}}  \sum_{x\sim y}\cbr{\1_{|x|\leq R }+\1_{|y|\leq R}}\abs{\trzero{a_x^\dagger a_y}}\label{Rem def}.
\end{align}

To bound the term in line \eqref{I def}, we apply the Cauchy–Schwarz inequality and make use of the finite-range nature of the interactions.
\begin{align}\label{I}
	(\mathrm{I})&\le s^{-1} \cbr{\sum_{x\sim y}\chi'_{ts}(x)\trzero{n_x}}^{1/2}\cbr{\sum_{x\sim y}\chi'_{ts}(y)\trzero{n_y}}^{1/2}\notag\\
	&= \frac{2d}{s}\trzero{\Nf[]{\chi'}{ts}}.
\end{align}
Similarly, we can control term in line \eqref{II def} as
\begin{align}
	(\mathrm{II})&\le\sum_{k=2}^\beta C'_{k}\,s^{-k} \cbr{\sum_{x\sim y}j'_{k,ts}(x)\trzero{n_x}}^{1/2}\cbr{\sum_{x\sim y} j'_{k,ts}(y)\trzero{n_y}}^{1/2}\notag\\
	&= 2d\sum_{k=2}^\beta C'_{k}\,s^{-k}\trzero{\Nf[]{j_k'}{ts}}.\notag
\end{align}
Assuming $R-r\ge 1$ and making use of the property of $\cE$, \eqref{prop C}, it follows that there exist $C''_{\beta}$ and $\tilde \chi\in\cE$ such that
\begin{align}\label{II}
	(\mathrm{II})&\le C''_{\be}s^{-2}\trzero{\Nf[]{\tilde \chi'}{ts}}.
\end{align}

Using the same ingredients as above, we estimate line \eqref{Rem def} as
\begin{align}\label{Rem}
	(\mathrm{Rem})&\le \frac{C_{\chi,\beta+1}}{s^{\beta+1}}  \cbr{\sum_{x\sim y}\cbr{\1_{|x|\leq R }+\1_{|y|\leq R}}\trzero{n_x}}^{1/2}\cbr{\sum_{x\sim y}\cbr{\1_{|x|\leq R }+\1_{|y|\leq R}}\trzero{n_y}}^{1/2}\notag\\
	&= \frac{C_{\chi,\beta+1}}{s^{\beta+1}}  \cbr{\sum_{\substack{x\in B_{R}\\y\;:\;x\sim y}}\trzero{n_x}+\sum_{\substack{y\in B_{R}\\x\;:\;x\sim y}}\trzero{n_x}}\notag\\
	&\le4d\,\frac{C_{\chi,\beta+1}}{s^{\beta+1}}  \trzero{N_{B_{R+1}}}.
\end{align}
Since the state $\rho$ was arbitrary, combining \eqref{I}, \eqref{II}, and \eqref{Rem} we obtain in the sense of quadratic forms on states in $\cD_1$
\begin{align}\label{bound on commutator}
	\mi\sbr{H,\Nf[]{\chi}{ts}}\le \frac{\kappa}{s}\Nf[]{\chi'}{ts}+ \frac{C }{s^{2}}\Nf[]{\tilde \chi'}{ts}+\frac{C}{s^{\beta+1}}  N_{B_{R+1}}
\end{align}
for some constant $C$ depending on $d,\chi, \beta$ and $J$.
Applying \eqref{first term derivative} and \eqref{bound on commutator} to \eqref{Heis der}, yields
\begin{align}
	D\Nf[]{\chi}{ts}\le  \frac{\kappa-\tilde v}{s}\Nf[]{\chi'}{ts}+ \frac{C }{s^{2}}\Nf[]{\tilde \chi'}{ts}+\frac{C}{s^{\beta+1}}  N_{B_{R+1}}.\notag
\end{align}
The inequality above, together with 
\begin{align}\label{Heis der and deriv}
	\od{}{t}\tev{\Nf[]{\chi}{ts}}{t}=\tev{D\Nf[]{\chi}{ts}}{t}
\end{align}
leads to the desired inequality.
\end{proof}

\subsubsection{Proof of \corref{cor boot}}

\begin{proof}
Consider \eqref{eq recurs stru}, integrating both sides with respect to time,	and applying the fundamental theorem of calculus we obtain
\begin{align}\label{int eq recurs stru}
	&\tev{\Nf[]{\chi}{ts}}{t}-\trzero{\Nf[]{\chi}{0s}}\notag\\
	&\qquad\le  \frac{\kappa-\tilde v}{s}\int_{0}^{t}\tev{\Nf[]{\chi'}{u s}}{u} \md u+\frac{C}{s^{2}}\int_{0}^{t} \tev{\Nf[]{\tilde \chi'}{u s}}{u} \md u+\frac{C}{s^{\beta}}\sup_{0\le u\le t}\tev{N_{B_{R+1}}}{u},
\end{align}
where the second integrated term in \eqref{int eq recurs stru} should be dropped for $\beta=1$. Above we also use that $s\ge t$.
Since the first summand on the r.h.s. of \eqref{eq recurs stru} is negative, it can be dropped. Then, after rearranging the terms it follows
\begin{align}\label{int eq drop}
	\tev{\Nf[]{\chi}{ts}}{t}\le  \trzero{\Nf[]{\chi}{0s}}+\frac{C}{s^{2}}\int_{0}^{t}\tev{\Nf[]{\tilde \chi'}{u s}}{u} \md u+\frac{C}{s^{\beta}}\sup_{0\le u\le t}\tev{N_{B_{R+1}}}{u}.
\end{align}
To conclude the proof we need to bound the integrated term in \eqref{int eq drop}. To this end, consider again \eqref{int eq recurs stru}, drop the first addend and rearrange the terms appropriately to obtain
\begin{align}\label{recurs ineq}
	\frac{1}{s}&\int_{0}^{t}\tev{\Nf[]{ \chi'}{u s}}{u} \md u\notag\\
    &\le C' \trzero{\Nf[]{\chi}{0s}} +\frac{C''}{s^{2}}\int_{0}^{t}\tev{\Nf[]{\tilde \chi'}{u s}}{u} \md u+\frac{C''}{s^\beta}\sup_{0\le u\le t}\tev{N_{B_{R+1}}}{u},
\end{align}
where the second integrated term in \eqref{int eq recurs stru} should be dropped for $\beta=1$.
The recursive inequality \eqref{recurs ineq} is a key feature of the ASTLO and can be applied repeatedly to control the second summand on the r.h.s. of \eqref{int eq drop}. 
This repeated application can be carried over since both functions $\chi$ and $\tilde \chi$ belong to the same function class $\cE$. 
Notice that at each iteration an additional factor $s$ is generated at denominator. 
Concretely, let us apply \eqref{recurs ineq} once, with $\beta-1$,  to control the integrated term in \eqref{int eq drop}. There exists a function $f\in\cE$ such that 
\begin{align}
	&\tev{\Nf[]{\chi}{ts}}{t}\notag\\
    &\quad\le  \trzero{\Nf[]{\chi}{0s}}+\frac{C}{s^{2}}\int_{0}^{t}\tev{\Nf[]{\tilde \chi'}{u s}}{u} \md u+\frac{C}{s^\beta}\sup_{0\le u\le t}\tev{N_{B_{R+1}}}{u}\notag\\
	&\quad\le  \trzero{\Nf[]{\chi}{0s}}+\frac{C_1}{s}\trzero{\Nf[]{\tilde \chi}{0s}}+\frac{C_1}{s^{3}}\int_{0}^{t}\tev{\Nf[]{f'}{u s}}{u} \md u+\frac{C_1}{s^\beta}\sup_{0\le u\le t}\tev{N_{B_{R+1}}}{u}.\notag
\end{align}
Applying repeatedly \eqref{recurs ineq} for $\beta-2,\, \beta-3, \dots, \,1$ to bound the integrated terms produced with every step, recalling \eqref{prop C}, and since $s\ge1$, we obtain 
\begin{align}
	\tev{\Nf[]{\chi}{ts}}{t}\le \cbr{1+\frac{C}{s}} \trzero{\Nf[]{\bar \chi}{0s}}+\frac{C}{s^\beta}\sup_{0\le u\le t}\tev{N_{B_{R+1}}}{u},\notag
\end{align}
for some $\bar \chi\in\cE$. 
Recalling the geometric properties of the ASTLO \eqref{f and B at t=0} yields the desired inequality \eqref{eq boot}.

\end{proof}

\subsubsection{Proof of \propref{prop remainder}}

\begin{proof}

Step 1.  Consider a large constant $L>0$ such that $\Lambda\subset B_L$. We start by showing \eqref{eq remainder control} for the special case of dyadic scales, namely
\begin{align}
	R=2^{k+1},\qquad r=2^k,\qquad k\ge0.\notag
\end{align}
Notice that $R-r= 2^k$.
Then, from inequality \eqref{eq boot} it follows
\begin{align}\label{from boot before multi}
	\tev{N_{B_{2^k}}}{t}&\le  \cbr{1+\frac{C}{2^k}} \trzero{N_{B_{2^{k+1}}}}+C\,2^{-\beta k}\sup_{0\le u\le t}\tev{N_{B_{2^{k+1}+1}}}{u}\notag\\
	&\le \cbr{1+\frac{C}{2^k}} \trzero{N_{B_{2^{k+1}}}}+C\,2^{-\beta k}\sup_{0\le u\le t}\tev{N_{B_{2^{k+2}}}}{u}.
\end{align}
To obtain \eqref{eq remainder control} we only need to show
\begin{align}\label{contr rem 1}
	2^{-dk}\sup_{0\le u\le t}\tev{N_{B_{2^{k}}}}{u}\le C_d \lambda,
\end{align}
for some constant $C_d$ to be determined later.
In fact, applying \eqref{contr rem 1} to \eqref{from boot before multi} yields 
\begin{align}
	\tev{N_{B_{2^k}}}{t}&\le  \cbr{1+\frac{C}{2^k}} \trzero{N_{B_{2^{k+1}}}}+C'2^{(d-\beta)k}\lam,\notag
\end{align}
where $C'=CC_d 2^{2d}$.
 We prove \eqref{contr rem 1} this via downward induction on $k$.

Step 1.1. We start by showing there exists $K>0$ large depending on $L$, such that \eqref{contr rem 1} holds for all $k\ge K$ for some constant $C_0$ independent of system size, $K$, and $k$. 
 Since the  Hamiltonian $H$ is number preserving it holds ${[N,H]=0}$. This, together with the assumption of controlled density \eqref{CD rho}, yields
\begin{align}
	2^{-d k}\sup_{0\le u\le t}\tev{N_{B_{2^{k}}}}{u}&\le 2^{-d k}\sup_{0\le u\le t}\tru{N}\notag\\
	&= 2^{-d k}\trzero{N}\notag\\
	&\le \lambda V_d\,2^{-d k} \,L^d\notag,
\end{align}
with $V_d$ the volume od the $d$ dimensional unit ball.
Then, for every $k\ge K:= \log_2L$,
\begin{align}
	2^{-d k}\sup_{0\le u\le t}\tev{N_{B_{2^{k}}}}{u}&\le C_0 \lam.\notag
\end{align}

Step 1.2. Now we show there exists $0\le\tilde K\le K$, independent of system size, such that \eqref{contr rem 1} holds for all $\tilde K\le k\le K$ for some constant $C_1$ independent on system size, $K$, and $k$. 
We prove the claim by induction on $k$. By point 1.1., the claim is proved for $K$ and $K+1$, and it constitutes the base of our induction. 
Let us assume \eqref{contr rem 1} holds for all $j>k$ with some constant $C_1$ to be determined, and let us show it holds for $k$. By \eqref{from boot before multi} it holds
\begin{align}\label{before ind hyp 2}
	2^{-d k}\sup_{0\le u\le t}\tev{N_{B_{2^{k}}}}{u}&\le2^{-dk} \cbr{\cbr{1+\frac{C}{2^{k}}}\trzero{N_{B_{2^{k+1}}}}+C \,2^{-\beta k}\sup_{0\le u\le t}\tev{N_{B_{2^{k+2}}}}{u}},
\end{align}
with $C$ depending on $\beta$, but not on $k$ and $K$.
We control the first summand in \eqref{before ind hyp 2} thanks the assumption of controlled density \eqref{CD rho},
\begin{align}\label{first term after ind 2}
	2^{-dk}\cbr{1+\frac{C}{2^{k}}}\trzero{N_{B_{2^{k+1}}}}\le (1+C)V_d2^d\,\lam\le \frac{C_1}{2}\lam,
\end{align}
where we chose $C_1:=\max\Set{C_0,2(1+C)V_d2^d}$.
To bound the second term in \eqref{before ind hyp 2} we apply \eqref{contr rem 1} for $k+2$,
\begin{align}
	C \,2^{-(\beta +d)k}\sup_{0\le u\le t}\tev{N_{B_{2^{k+2}}}}{u}\le CC_1 2^{2d}  2^{-\beta k}\lam.\notag
\end{align}
Choosing $k\ge \beta^{-1}\log_2(C2^{2d+1})$, it holds
\begin{align}\label{exp val after ind 2}
	C \,2^{-(\beta +d)k}\sup_{0\le u\le t}\tev{N_{B_{2^{k+2}}}}{u}\le \frac{C_1}{2}\lam.
\end{align}
Applying \eqref{first term after ind 2} and \eqref{exp val after ind 2} to \eqref{before ind hyp 2}, we obtain
\begin{align}
	2^{-d k}\sup_{0\le u\le t}\tev{N_{B_{2^{k}}}}{u}&\le  C_1\lam.\notag
\end{align}
Thus, \eqref{contr rem 1} holds for all $k\ge \tilde K:= \beta^{-1}\log_2(C2^{2d+1})$ for $C_1$ defined in the above.

Step 1.3. To conclude the proof for dyadic scales, we show \eqref{contr rem 1} holds for all $0\le k< \tilde K$ for some $C_2$ independent of systems size. 
Consider $k=\tilde K-1$, by the same reasoning as in step 1.2., and knowing \eqref{contr rem 1} holds for $j\ge \tilde K$ with $C_1$, we can write
\begin{align}
	2^{-d (\tilde K-1)}\sup_{0\le u\le t}\tru{N_{B_{2^{\tilde K-1}}}}&\le  \cbr{ (1+C)V_d2^d + CC_1 2^{2d} 2^{-\beta(\tilde K-1)}}\lam.\notag
\end{align}
Thus, \eqref{contr rem 1} holds for $\tilde K$ for constant $C_{2,1}:=\max\Set{C_1, (1+C)V_d2^d + CC_1 2^{2d} 2^{-\beta }}$.
We repeat this reasoning for all $k=\tilde K-2,\dots,0$, and, since $\tilde K$ is independent of system size, this shows \eqref{contr rem 1} holds for all $k\ge 0$ with constant 
\begin{align}
	{C_2:=\max\Set{ C_{2,1},\dots , C_{2,\tilde K}}}.	\notag
\end{align}

Step 2. Now we adapt the proof to $R,r\ge 0$ with $R-r\ge1$. We build an induction on $R-r>2^k$, $k=1,2,\dots$. 

Step 2.1. We start by setting the base case as in step 1.1.
As before, we only need to show 
\begin{align}\label{need rem contr}
	s^{-d}\sup_{0\le u\le t}\tev{N_{B_{R+1}}}{u}&\le C_d\lambda,
\end{align}
for some $C_d>0$.
As before, we can obtain the following generous bound
\begin{align}
	s^{-d}\sup_{0\le u\le t}\tev{N_{B_{R+1}}}{u}
	&\le s^{-d}\sup_{0\le u\le t}\tru{N}\notag\\
	&=s^{-d}\trzero{N}\notag\\
	&\le\lam V_d\,\cbr{\frac{L}{R-r}}^d.\notag
\end{align}
Thus, for $K=\log_2 L$ and all $k\ge K$, we have the desired estimate, 
\begin{align}
	s^{-d}\sup_{0\le u\le t}\tev{N_{B_{R+1}}}{u}&\le C_0\lambda,\notag
\end{align} 
for ${R}-r\ge 2^k$ and $C_0>0$.

Step 2.2. As a next step we show there exists $0\le \tilde K\le K$, independent of $L$, such that \eqref{need rem contr} holds for some constant $C_1$ independent of system size, $k$, and $\tilde K$. 
Let us assume \eqref{need rem contr} holds for all  $R', r'$ such that ${R'-r'>2^{k+1}}$ for $\tilde K<k\le K$ and $C_1$ where $\tilde K$ and $C_1$ are to be determined later.
We apply \eqref{eq boot}  for $r'=R+1$, $R'=3R-2r+1$ to control the time average appearing in \eqref{need rem contr}.
\begin{align}\label{before ind hyp}
	(R-r)^{-d}&\sup_{0\le u\le t}\tev{N_{B_{R+1}}}{u}\notag\\
	&\le (R-r)^{-d}\cbr{1+\frac{C}{2(R-r)}}\trzero{N_{B_{3R-2r+1}}}+\frac{C}{ 2^\beta(R-r)^{\beta+d}}\sup_{0\le u\le t}\tru{N_{B_{3R-2r+2}}},
\end{align}
for $C$ independent of $k,K$.
We control the first term in \eqref{before ind hyp} using the assumption of controlled density \eqref{CD rho}
\begin{align}
	(R-r)^{-d}\cbr{1+\frac{C}{2(R-r)}}\trzero{N_{B_{3R-2r+1}}}&\le \cbr{1+C}V_d\cbr{\frac{3R-2r+1}{R-r}}^d\lam\notag\\
	&\le \cbr{1+C}V_d\cbr{\frac{4R}{R-r}}^d\lam.\notag
\end{align}
Since we assumed $R-r\ge \delta_0 r$ it holds $R-r\ge \mu_0 R$ with $\mu_0:=1-\frac{1}{1+\delta_0}>0$.  Thus, by the inequality above it follows
\begin{align}\label{first after ind r}
	(R-r)^{-d}\cbr{1+\frac{C}{2(R-r)}}\trzero{N_{B_{3R-2r+1}}}&\le  \cbr{1+C}V_d\cbr{4\mu_0^{-1}}^d\lam\notag\\
	&\le \frac{C_1}{2}\lam,
\end{align}
with
\begin{align}
	C_1:=\max\Set{C_0, 2\cbr{1+C}V_d\cbr{4\mu_0^{-1}}^d }.\notag
\end{align}
To control the second summand in \eqref{before ind hyp} we apply the induction hypothesis \eqref{need rem contr} to obtain 
\begin{align}\label{exp after ind r}
	\frac{C}{ 2^\beta(R-r)^{\beta+d}}\sup_{0\le u\le t}\tru{N_{B_{3R-2r+2}}}&\le CC_1 2^{2d-\beta}(R-r)^{-\beta}\lam\notag\\
	&\le \frac{C_1}{2}\lam,
\end{align}
where we considered $R-r>2^k$ for $k\ge\tilde K:=\beta^{-1}\log_2( C2^{2d-\beta+1})$.
Inequality \eqref{before ind hyp}, together with \eqref{first after ind r} and \eqref{exp after ind r}, yields
\begin{align}
	(R-r)^{-d}&\sup_{0\le u\le t}\tev{N_{B_{R+1}}}{u}\le C_1 \lam\notag
\end{align}
for all $\tilde K \le k\le K$.

Step 2.3. Now we show \eqref{need rem contr} holds for all $R,r$ such that $R-r\ge 2^k$ for all $0\le k < \tilde K$. 
Let us first consider $k=\tilde K-1$. By the same argument as above we obtain
\begin{align}
	(R-r)^{-d}&\sup_{0\le u\le t}\trt{N_{B_{R+1}}}\notag\\
	&\le \cbr{\cbr{1+C}V_d\cbr{4\mu_0^{-1}}^d+\frac{CC_12^{2d-\beta}}{(R-r)^{\beta}}}\lam\notag\\
	&\le \cbr{\cbr{1+C}V_d\cbr{4\mu_0^{-1}}^d+CC_12^{2d-\beta}}\lam.\notag
\end{align}
This yields \eqref{need rem contr} with constant $C_{2,1}:=\max\Set{C_1,\cbr{1+C}V_d\cbr{4\mu_0^{-1}}^d+CC_12^{2d-\beta} }$.
Repeating this reasoning for $0\le k\le \tilde K-1$ we obtain \eqref{need rem contr} with constant $C_{2}:=\max\Set{C_{2,1},\dots,C_{2,\tilde K}}$.
\end{proof}

\subsection{Proof of the higher moment bound}\label{sec proof high PPB}

In this section we prove \propref{prop recurs stru p}, \corref{cor boot p}, and \propref{prop rem p} by induction on the moment parameter $\eta $, where the results in \secref{sec mom 1} constitute the basis of the induction.
The proofs follow the proof strategy of \cite{LRZ2023}. 
Recall that in what follows the constant are allowed to change from line to line still remaining independent of system size.

\subsubsection{Proof of \propref{prop recurs stru p}}

\begin{proof}
	We start computing the Heisenberg derivative of $\Nf[\eta]{\chi}{ts}$ with respect to time. 
	\begin{align}\label{after der p}
		D\Nf[\eta]{\chi}{ts}
		= &-\eta \frac{\tilde v}{s}N_{\chi,ts}^{\eta-1}N_{\chi',ts}
		+[\mi H,\Nf[\eta]{\chi}{ts}] .
	\end{align}
	We need to bound the second summand in \eqref{after der p} using a similar procedure we used in the proof of \propref{prop recurs stru}. By Leibniz rule it holds
	\begin{align}
		[\mi H,\Nf[\eta]{\chi}{ts}] &=\sum_{\zeta=0}^{\eta-1} N_{\chi,ts}^\zeta[\mi H,N_{\chi,ts}]N_{\chi,ts}^{\eta-\zeta-1}\notag\\
		&=J\sum_{x\sim y}\sum_{\zeta=0}^{\eta-1} N_{\chi,ts}^\zeta[\mi a_x^\dagger a_y,N_{\chi,ts}]N_{\chi,ts}^{\eta-\zeta-1}.\notag
	\end{align}
	The following commutation relations hold
	\begin{align}
		\label{commutator}
		&[N_{\chi,ts},a_x^\dagger]=\sum_z \chi_{ts}(z) [n_z,a_x^\dagger]
		=\sum_z \chi_{ts}(z) a_z^\dagger \delta_{z,x}
		=\chi_{ts}(x) a_x^\dagger,\\
		&[a_y, N_{\chi,ts},]=\sum_z \chi_{ts}(z) [a_y,n_z]
		=\sum_z \chi_{ts}(z) \delta_{y,z}a_z
		=\chi_{ts}(y) a_y.\label{commutator'}
	\end{align}
Then we can write
\begin{align}
	N_{\chi,ts}^\zeta[\mi a_x^\dagger a_y,N_{\chi,ts}]N_{\chi,ts}^{\eta-\zeta-1}=\mi \cbr{\chi_{ts}(x)-\chi_{ts}(y)}N_{\chi,ts}^\zeta a_x^\dagger a_y N_{\chi,ts}^{\eta-\zeta-1}.\notag
\end{align}
Thus, for a given  $\rho\in\cD_\eta $,
\begin{align}\label{after comm}
	\abs{\trzero{[\mi H,\Nf[\eta]{\chi}{ts}]}}\le \abs{J}\sum_{x\sim y}\abs{\chi_{ts}(x)-\chi_{ts}(y)}\sum_{\zeta=0}^{\eta-1}\abs{\trzero{ N_{\chi,ts}^\zeta a_x^\dagger a_yN_{\chi,ts}^{\eta-\zeta-1}}}.
\end{align}
As in the proof of \propref{prop recurs stru} we aim to apply Cauchy-Schwarz to the r.h.s. of \eqref{after comm} in order to recover the ASTLOs. In order to obtain a symmetric expression we first need to move the existing ASTLO in between the creation and annihilation operator.
To this end, we make use of the following relations,
\begin{align}
	N_{\chi,ts}^\zeta a_x^\dagger &=a_x^\dagger (N_{\chi,ts}+\chi_{ts}(x))^\zeta\label{NaRel}\\
	a_y N_{\chi,ts}^{\eta-\zeta-1}&=(N_{\chi,ts}+\chi_{ts}(y))^{\eta-\zeta-1} a_y.\label{NaRel'}
\end{align}
Notice \eqref{NaRel}--\eqref{NaRel'} can be proven by induction with \eqref{commutator}--\eqref{commutator'} as a base case.
We define 
\begin{align}
	A(x):=N_{\chi,ts}+\chi_{ts}(x),\notag
\end{align}
and \eqref{after comm}, \eqref{NaRel}, and \eqref{NaRel'} yield
\begin{align}
	\abs{\trzero{[\mi H,\Nf[\eta]{\chi}{ts}]}}\le \abs{J}\sum_{\zeta=0}^{\eta-1}\sum_{x\sim y}\abs{\chi_{ts}(x)-\chi_{ts}(y)}\abs{\trzero{a_x^\dagger A(x)^\zeta A(y)^{\eta-\zeta-1} a_y}}.\notag
\end{align}
To control the difference in the above inequality we make use of \lemref{lemm sym taylor} as we did in the previous section
\begin{align}
	\abs{\trzero{[\mi H,\Nf[\eta]{\chi}{ts}]}}&\le \abs{J}\sum_{\zeta=0}^{\eta-1}\Big(s^{-1}\sum_{x\sim y}u_{ts}(x)u_{ts}(y)\abs{\trzero{a_x^\dagger A(x)^\zeta A(y)^{\eta-\zeta-1} a_y}}\label{after taylo p 1}\\
	&\quad + \sum_{k=2}^{\beta}C_{\chi,k}s^{-k}\sum_{x\sim y}u_{k,ts}(x)u_{k,ts}(y)\abs{\trzero{a_x^\dagger A(x)^\zeta A(y)^{\eta-\zeta-1} a_y}}\label{after taylo p 2}\\
	&\quad+ C_{\chi,\beta}s^{-\beta-1}\sum_{x\sim y}\cbr{\1_{|x|\leq R }+\1_{|y|\leq R}}\abs{\trzero{a_x^\dagger A(x)^\zeta A(y)^{\eta-\zeta-1} a_y}}\Big).\label{after taylo p 3}
\end{align}
Let us focus on \eqref{after taylo p 1}.
Since $A(z)$ are positive commuting operators and that $\chi\le 1$ it holds 
\begin{align}\label{Ax Ay < N}
    A(x)^\zeta A(y)^{\eta-\zeta-1}\le (N_{\chi,ts}+1)^{\eta-1}
\end{align}
We apply Cauchy-Schwarz and \eqref{Ax Ay < N} to control the the trace in \eqref{after taylo p 1} 
\begin{align}\label{term 1}
	&\sum_{\zeta=0}^{\eta-1}\sum_{x\sim y}u_{ts}(x)u_{ts}(y)\abs{\trzero{a_x^\dagger A(x)^\zeta A(y)^{\eta-\zeta-1} a_y}}\notag\\
	& \le \sum_{\zeta=0}^{\eta-1}\cbr{\sum_{x\sim y}\chi'_{ts}(x)\trzero{a_x^\dagger A(x)^\zeta A(y)^{\eta-\zeta-1}a_x}}^{1/2}\cbr{\sum_{x\sim y}\chi'_{ts}(y)\trzero{a_y^\dagger A(x)^\zeta A(y)^{\eta-\zeta-1}a_y}}^{1/2}\notag\\
	&\le2d\eta\,\sum_{x\in\Lam}\chi'_{ts}(x)\trzero{a_x^\dagger (N_{\chi,ts}+1)^{\eta-1} a_x}.
\end{align}
To recover the ASTLO we apply once again \eqref{NaRel'} to obtain
\begin{align}\label{get ASTTLO}
	\abs{\trzero{a_x^\dagger (N_{\chi,ts}+1)^{\eta-1} a_x}}&=\abs{\trzero{n_x (N_{\chi,ts}+1-\chi_{ts}(x))^{\eta-1} }}\notag\\
	&\le\trzero{n_x (N_{\chi,ts}+1)^{\eta-1} }.
\end{align}
Combining \eqref{term 1} and \eqref{get ASTTLO} yields
\begin{align}\label{term 1 1}
	\sum_{\zeta=0}^{\eta-1}\sum_{x\sim y}u_{ts}(x)u_{ts}(y)\abs{\trzero{a_x^\dagger (N_{\chi,ts}+1)^{\eta-1} a_y}}&\le 2d\eta\,\sum_{x\in\Lam}\chi'_{ts}(x)\trzero{n_x (N_{\chi,ts}+1)^{\eta-1} }\notag\\
&=2d\eta\,\trzero{N_{\chi',ts} (N_{\chi,ts}+1)^{\eta-1} }.
\end{align}
Analogously we can bound line \eqref{after taylo p 2} as 
\begin{align}
	\sum_{\zeta=0}^{\eta-1}\,\sum_{k=2}^{\beta}C_{\chi,k}s^{-k}&\sum_{x\sim y}u_{k,ts}(x)u_{k,ts}(y)\abs{\trzero{a_x^\dagger  A(x)^\zeta A(y)^{\eta-\zeta-1} a_y}}\notag\\
	&\le 2d \eta\sum_{k=2}^{\beta}C_{\chi,k}s^{-k}\sum_{x\sim y}\trzero{N_{j_k',ts} (N_{\chi,ts}+1)^{\eta-1} }\notag.
\end{align}
Additionally, thanks to \eqref{prop C} and $R-r>1$, there exists $\tilde \chi\in\cE$ such that 
\begin{align}
	\sum_{\zeta=0}^{\eta-1}\,\sum_{k=2}^{\beta}C_{\chi,k}s^{-k}&\sum_{x\sim y}u_{k,ts}(x)u_{k,ts}(y)\abs{\trzero{a_x^\dagger  A(x)^\zeta A(y)^{\eta-\zeta-1} a_y}}\notag\\
	&\le 2d \eta\,C_{\chi,\beta} \,s^{-2}\trzero{N_{\tilde \chi',ts} (N_{\chi,ts}+1)^{\eta-1} }\notag
\end{align}
Applying Cauchy-Schwarz and \eqref{get ASTTLO} we can   control \eqref{after taylo p 3} as 
\begin{align}\label{rem p}
	\sum_{\zeta=0}^{\eta-1}\,\sum_{x\sim y}\cbr{\1_{|x|\leq R }+\1_{|y|\leq R}}\abs{\trzero{a_x^\dagger A(x)^\zeta A(y)^{\eta-\zeta-1} a_y}}&\le 2d \eta\sum_{x\in B_{R+1}}\trzero{n_x (N_{\chi,ts}+1)^{\eta-1} }\notag\\
	&= 2d\eta\, \,\trzero{N_{ B_{R+1}} (N_{\chi,ts}+1)^{\eta-1} }\notag\\
	&\le 2d \eta\,\trzero{N_{ B_{R+1}} (N_{ B_{R+1}}+1)^{\eta-1} }.
\end{align}
Notice that in the last step we have used \lemref{prop geom prop}.
Bounding \eqref{after taylo p 1}--\eqref{after taylo p 3} thanks to \eqref{term 1 1}--\eqref{rem p} leads to
\begin{align}\label{bound comm p}
	\abs{\trzero{[\mi H,\Nf[\eta]{\chi}{ts}]  }}\le 2d \eta \abs{J}\Big(&s^{-1}\trzero{N_{\chi',ts} (N_{\chi,ts}+1)^{\eta-1} }\notag\\
	&+ C_{\chi,\beta}\;s^{-2}\trzero{N_{\tilde \chi',ts} (N_{\chi,ts}+1)^{\eta-1} }\notag\\
	&+C_{\chi,\beta}\;\trzero{N_{ B_{R+1}} (N_{ B_{R+1}}+1)^{\eta-1} }\Big).
\end{align}
Inequality \eqref{bound comm p} together with \eqref{after der p} yields, after rearranging the various terms appropriately, 
\begin{align}\label{almost final}
	\trzero{D\Nf[\eta]{\chi}{ts}}\le & \;\eta\tilde v  s^{-1}\trzero{N_{\chi',ts}\cbr{(N_{\chi,ts}+1)^{\eta-1}-N_{\chi,ts}^{\eta-1}}}\notag\\
	&+\eta(\kappa -\tilde v) s^{-1}\trzero{N_{\chi',ts}(N_{\chi,ts}+1)^{\eta-1} }
	\notag\\
	&+\eta C_{\chi,\beta}\;s^{-2}\trzero{N_{\tilde \chi',ts} (N_{\chi,ts}+1)^{\eta-1}}\notag\\
	&+\frac{ \eta C_{\chi,\beta}}{s^{\beta+1}} \trzero{N_{B_{R+1}}(N_{B_{R+1}}+1)^{\eta-1}}.
\end{align}
Notice that for $h\ge0$, 
		\begin{align}
			\label{MVT}
			(h+1)^{\eta-1}-h^{\eta-1}\leq (\eta-1)(h+1)^{\eta-2}.
		\end{align}
Applying \eqref{MVT} to \eqref{almost final} and recalling $\rho$ was arbitrary, we obtain
\begin{align}
	D\Nf[\eta]{\chi}{ts}\le &\;\eta(\eta-1) \tilde v s^{-1}\,N_{\chi',ts}(N_{\chi,ts}+1)^{\eta-2}\notag\\
	&+\eta(\kappa -\tilde v) s^{-1}\,N_{\chi',ts}(N_{\chi,ts}+1)^{\eta-1} 
	\notag\\
	&+\eta C_{\chi,\beta}\;s^{-2}\,N_{\tilde \chi',ts} (N_{\chi,ts}+1)^{\eta-1}\notag\\
	&+\frac{ \eta C_{\chi,\beta}}{s^{\beta+1}} \,N_{B_{R+1}}(N_{B_{R+1}}+1)^{\eta-1}\notag,
\end{align}
in the sense of quadratic forms on $\cD_\eta $.
Given $\rho\in\cD_\eta $, combining the inequality above with \eqref{Heis der and deriv} we obtain the desired bound \eqref{eq recurs stru p}

\end{proof}

\subsubsection{Proof of \corref{cor boot p}}

\begin{proof}
We start by showing, via induction on $\eta$, the following inequality
\begin{align}\label{induction on p}
	\int_0^t\tru{N_{\chi',u s}(N_{\chi,u s}+1)^{\eta-1} } \md u&\le  C \sum_{\zeta=1}^\eta \Big(s \cbr{\trzero{\Nf[\zeta]{\chi}0}-\trt{\Nf[\zeta]{\chi}{ts}}}\notag\\
&\quad +s^{-1}\int_0^t\tru{N_{\tilde\chi',u s}(N_{\chi,u s}+1)^{\zeta-1} } \md u\notag\\
&\quad+\frac{ t}{s^{\beta}}\sup_{0\le u\le t} \tru{N_{B_{R+1}}(N_{B_{R+1}}+1)^{\zeta-1}}\Big) ,
\end{align} 
where the second summand on the r.h.s. of \eqref{induction on p} should be dropped for $\beta=1$.
Inequality \eqref{int eq recurs stru} proves the base case $\eta=1$. 
Now we assume \eqref{induction on p} holds for $\eta-1$ and we show it for $\eta$.
To this end, we integrate inequality \eqref{eq recurs stru p}, rearrange the terms, and divide both sides by $\eta(\kappa -\tilde v) s^{-1}$ to obtain
\begin{align}\label{int}
\int_0^t\tru{N_{\chi',u s}(N_{\chi,u s}+1)^{\eta-1} } \md u&\le  C (\eta-1)\int_0^t\tru{N_{\chi',u s}(N_{\chi,u s}+1)^{\eta-2} } \md u\notag\\
&\quad +C s \cbr{\trzero{\Nf[\eta]{\chi}0ts}-\trt{\Nf[\eta]{\chi}{ts}}}\notag\\
&\quad +Cs^{-1}\int_0^t\tru{N_{\tilde\chi',u s}(N_{\chi,u s}+1)^{\eta-1} } \md u\notag\\
&\quad+\frac{C }{s^{\beta-1}}\sup_{0\le u\le t} \tru{N_{B_{R+1}}(N_{B_{R+1}}+1)^{\eta-1}} .
\end{align}
We also made use of the fact $s\ge t$.
Notice that the first term on the r.h.s. of \eqref{int} is exactly of the form of the term on the l.h.s. of \eqref{induction on p} with $\eta-1$. Thus, thanks to the induction hypothesis \eqref{induction on p} follows.
Thanks to the properties of the function class $\cE$ and applying iteratively \eqref{induction on p} to control the integrated term appearing on the r.h.s. of \eqref{induction on p} we obtain
\begin{align}\label{bound on int}
	\int_0^t&\tru{N_{\chi',u s}(N_{\chi,u s}+1)^{\eta-1} } \md u\notag\\
    &\le  C \sum_{\zeta=1}^\eta \Big(s \cbr{\trzero{\Nf[\zeta]{\chi}0}-\trt{\Nf[\zeta]{\chi}{ts}} }+\trzero{\Nf[\zeta]{\bar \chi}0}-\trt{\Nf[\zeta]{\bar\chi}{ts}} \notag\\
&\quad+\frac{ 1}{s^{\beta-1}}\sup_{0\le u\le t} \tru{N_{B_{R+1}}(N_{B_{R+1}}+1)^{\zeta-1}}\Big)
\end{align}
for some $\bar \chi\in\cE$.
We drop the integrated term in \eqref{bound on int}, rearrange the terms, and divide both sides by $Cs$, to obtain
\begin{align}\label{before final}
	\trt{\Nf[\eta]{\chi}{ts}}&\le \sum_{\zeta=1}^\eta \cbr{\trzero{\Nf[\zeta]{\chi}0} +s^{-1}\trzero{\Nf[\zeta]{\bar \chi}0}-s^{-1}\trt{\Nf[\zeta]{\bar\chi}{ts}} }-\sum_{\zeta=1}^{\eta-1}\trt{\Nf[\zeta]{\chi}{ts}}\notag\\
	&\quad+ \frac{ C'}{s^{\beta}}\sum_{\zeta=1}^\eta\sup_{0\le u\le t}\tru{N_{B_{R+1}}(N_{B_{R+1}}+1)^{\zeta-1}}\notag \\
	&\le \sum_{\zeta=1}^\eta \Big( \cbr{\trzero{\Nf[\zeta]{\chi}0} +s^{-1}\trzero{\Nf[\zeta]{\bar \chi}0} }+\frac{ C'}{s^{\beta}}\sup_{0\le u\le t} \tru{N_{B_{R+1}}(N_{B_{R+1}}+1)^{\zeta-1}}\Big).
\end{align}
Applying \lemref{prop geom prop}, the algebraic identity $\sum_{\zeta=1}^\eta (1+h)^{\zeta-1}\le C_\eta (1+h^{\eta-1})$, and the fact that the number operator has integers eigenvalues, it follows
\begin{align}\label{bound power}
	\sum _{\zeta=1}^\eta  N_{ f,0s}^\zeta  \le 	\sum _{\zeta=1}^\eta  (N_{ f,0s}+1)^{\zeta-1}N_{ f,0s}\le C_\eta \cbr{ N_{ f,0s}+N_{ f,0s}^\eta}\le C_\eta (N_{B_R}+N^\eta_{B_R})\le 2C_\eta N^\eta _{B_R},
\end{align}
for $f\in\cE$.
\lemref{prop geom prop} and \eqref{bound power} together allow us to derive from \eqref{before final}
\begin{align}
	\trt{N^\eta _{B_r}}&\le C\cbr{1+s^{-1}}\trzero{N^\eta _{B_R}}+ \frac{ C}{s^{\beta}}\sup_{0\le u\le t}\tru{N_{B_{R+1}}(N^{\eta-1}_{B_{R+1}}+1)}\notag\\
	&\le C'\trzero{N^\eta _{B_R}}+ \frac{ C'}{s^{\beta}}\sup_{0\le u\le t}\tru{N^{\eta}_{B_{R+1}}}.\notag
\end{align}

\end{proof}

\subsubsection{Proof of \propref{prop rem p}}

\begin{proof}
 Step 1. As in the proof of \propref{prop remainder} we start by showing the desired inequality for the special case of dyadic scales via a downwards induction on the scale.

Step 1.1. We consider $R=2^{k+1},\,r=2^k$, and we first show \eqref{eq contr rem} for all $k\ge K$ for some $K(L)>0$ with $\Lambda\subset B_L$.
Notice that we only need to show 
\begin{align}\label{need for rem}
	2^{-d\eta k}\sup_{0\le u\le t}\tru{N^{\eta}_{B_{2^{k}}}}\le C_{d,\eta}\lam^\eta.
\end{align}
To this end we consider \eqref{eq boot p}, we generously upperbound the remainder term with the total number operator, then, thanks to $[H,N_{\Lambda}]=0$ and the assumption of controlled density \eqref{CD rho}, we obtain
\begin{align}
	2^{-d\eta k}\sup_{0\le u\le t}\tru{N^{\eta}_{B_{2^{k}}}}\le 2^{-d\eta k}\sup_{0\le u\le t}\trzero{N^{\eta}}\le \cbr{ \frac{\lam V_d L^{d}}{2^{dk}}}^\eta,\notag
\end{align}
with $V_d$ the volume of the $d$ dimensional unit ball.
Thus, for every $k>K:=\log_2L$ \eqref{need for rem} holds with some constant $C_0>0$.

Step 1.2. In this step we show there exists $\tilde K\le K$ independent of system size, such that \eqref{need for rem} holds with a constant $C_1$ for all $\tilde K\le k< K$ .
Now assume \eqref{need for rem} holds for all $j>k$, and show it implies \eqref{eq contr rem} for $k$. Inequality \eqref{eq boot p}, the assumption of controlled density, and the induction hypothesis yield
\begin{align}
	2^{-d\eta k}&\sup_{0\le u\le t}\tru{N^\eta _{B_{2^{k}}}}\notag\\
    &\le 2^{-d\eta k}\cbr{C \trzero{N^\eta _{B_{2^{k+1}}}}+C'2^{-\beta k}\sup_{0\le u\le t}\tru{N^\eta _{B_{2^{k+2}}}}}  \notag\\
	&\le \cbr{CV_d^\eta 2^{d\eta }+C'C_12^{2d\eta -\beta k}}\lam^\eta .\notag
\end{align}
Assuming  $k\ge \tilde K:=\beta^{-1}\log_2(C'2^{2d\eta+1})$ and setting $C_1:=\max\Set{C_0, CV_d^\eta 2^{d\eta+1 }}$, it follows
\begin{align}
	2^{-d\eta k}\tru{N^\eta _{B_{2^{k}}}}&\le C_1\lam^\eta.\notag
\end{align}
This shows \eqref{need for rem} holds for all $k\ge \tilde K$ with constant $C_1$.

Step 1.3. In order to show  \eqref{need for rem} holds for all $0\le k < \tilde K$ with some constant $C_2$ consider $k=\tilde K-1$.
By the above argument it holds
\begin{align}
	2^{-d\eta k}\sup_{0\le u\le t}\tru{N^\eta _{B_{2^{k}}}}	&\le \cbr{CV_d^\eta 2^{d\eta }+C'C_12^{2d\eta -\beta k}}\lam^\eta .\notag
\end{align}
Thus, for $k=\tilde K-1$, \eqref{need for rem} holds with constant $C_{2,1}:=\max\Set{C_1,CV_d^\eta 2^{d\eta }+C'C_12^{2d\eta -\beta }}$. 
Repeating this procedure for all $0\le k\le \tilde K-2$, we obtain \eqref{need for rem} for all $k\ge0$ with constant
\begin{align}
	C_2:=\max\Set{C_{2,1},\dots,C_{2,\tilde K}}.\notag
\end{align}
This closes the induction and shows \eqref{eq contr rem} holds for dyadic scales. 

Step 2. The proof for general $R,r\ge 0$ with $R-r\ge \delta_0r$  follows directly by generalizing to higher $\eta$ the strategy we utilized in the proof of \propref{prop remainder}.

\end{proof}

\section{Deriving Lieb--Robinson bounds}\label{sec derive LRB}
In this section, using the propagation bounds developed in \thmref{theo final particle}, we  prove \thmref{teo LRB}.
For every $Y\subset\Lam$, $\nu>0$, we define the following projectors
\begin{align}
	\Pi_{Y,\nu} := \prod_{x\in Y}\1_{n_x\le \nu}, \qquad \Pi_{Y,\nu}^\perp	:=\1-\Pi_{Y,\nu}\notag
\end{align}
 Notice that 
\begin{align}\label{prop p perp}
	\Pi_{Y,\nu}^\perp \le \sum_{x\in Y}\1_{n_x> \nu},\qquad \Pi_{Y,\nu}^\perp\, \me^{\mi t \bar H}=\Pi_{Y,\nu}^\perp .
\end{align}
Given an operator $A$ acting on Fock space, we write
\begin{align}\label{bar def}
	\bar A:= \Pi_{Y,\nu} A\, \Pi_{Y,\nu}.
\end{align}
We observe that, given two operators $A,B$ acting on bosonic Fock space that are supported on sets $X,Y$, respectively, with $X\cap Y=\emptyset$, it holds that $\sbr{\bar A, \bar B}=0$.
Consider the following truncated dynamics
\begin{align}\label{def tr dyn}
	\bar \tau_t (A):=\me^{\mi t \bar H }A\,\me^{-\mi t \bar H }.
\end{align}
Fix $R>2$ and consider $\Pi=\Pi_{X[R+2],\nu}$. We write  as $\bar \tau_t^R(A)$ the dynamics generated by $\bar H_{X[R]}$ and by $\tau_t^R(A)$ the one generated by $ H_{X[R]}$.
The proof to approximate $\tau_t (A)$ by $\tau^R_t (A)$ follows the steps below
\begin{align}\label{strategy LRB}
	\tau_t (A)\,\xrightarrow{\hspace{.8em}(1)\hspace{.8em}} \,\tau_t (\bar A)\,\xrightarrow{\hspace{.8em}(2)\hspace{.8em}} \,\bar\tau_t (\bar A)\,\xrightarrow{\hspace{.8em}(3)\hspace{.8em}} \,\bar\tau^R_t (\bar A)\,\xrightarrow{\hspace{.8em}(4)\hspace{.8em}} \,\tau^R_t (\bar A)\,\xrightarrow{\hspace{.8em}(5)\hspace{.8em}} \,\tau^R_t (A).
\end{align}
The following lemma allows us to complete steps (1) and (5).
\begin{lemma}\label{lemma bar A}
Fix $\eta\ge 1$. There exists a constant $C=C(J,d,\eta)>0$ such that, for any $X\subset Y\subset\Lambda$ and $\nu>0$, all  $A\in\Ainv{X}$ and $B\in\cB$, and any state $\rho\in \cD_\eta$ such that \eqref{CD rho} holds for $\eta$ and some $\lam>0$, the following holds
	\begin{align}
		\abs{\Tr\sbr{	 \tau_t \cbr{A-\bar A}B\;\rho}}\le C \norm{A}\norm{B}\abs{Y}^{1/2}d(X)^{d\eta /2}\cbr{\frac{\lam \cbr{\abs{t}+1}^{d}}{\nu}}^{\eta/2} ,\notag
	\end{align}
    for any $t\in \R$.
\end{lemma}
\begin{proof}
Since $\1=\Pi+\Pi^\perp$, it holds
	\begin{align}
		A-\bar A=\Pi^\perp A\Pi+\Pi A\Pi^\perp+\Pi^\perp A\Pi^\perp \notag.
	\end{align}
By triangular inequality, the equation above implies
\begin{align}
    \abs{\Tr\sbr{	 \tau_t \cbr{A-\bar A}B\;\rho}}&\le \abs{\Tr\sbr{	 \tau_t \cbr{\Pi^\perp A\Pi}B\;\rho}}\label{1}\\
    &+ \abs{\Tr\sbr{	 \tau_t \cbr{\Pi A\Pi^\perp}B\;\rho}}\label{2}\\
    &+ \abs{\Tr\sbr{	 \tau_t \cbr{\Pi^\perp A\Pi^\perp}B\;\rho}}\label{3}
\end{align}
Thanks to Cauchy-Schwarz, line \eqref{3} can be bounded as
	\begin{align}
		\abs{\Tr\sbr{\rho\;	 \me^{\mi t H} \Pi^\perp A\Pi^\perp\me^{-\mi t H}B}}\le \norm{A}\norm{B}\abs{\Tr\sbr{\me^{-\mi t H}\rho\;	 \me^{\mi t H} \Pi^\perp }}^{1/2}.\notag
	\end{align}
    From Markov's inequality and \thmref{theo final particle} with $v=4d\abs{J}$ and $\delta_0=1/2$, it follows
    \begin{align*}
         \abs{\Tr\sbr{\me^{-\mi t H}\rho\;	 \me^{\mi t H} \Pi^\perp }}\le C \abs{Y}\cbr{\frac{\lam \abs{t}^d}{\nu}}^{\eta}.
    \end{align*}
    Thus,
    \begin{align}
        \abs{\Tr\sbr{\rho\;	 \me^{\mi t H} \Pi^\perp A\Pi^\perp\me^{-\mi t H}B}}\le C\norm{A}\norm{B} \sqrt{\abs{Y}}\cbr{\frac{\lam \abs{t}^d}{\nu}}^{\eta/2},
    \end{align}
    for some constant $C=C(d,J,\eta)>0.$
	Similarly, the term on line \eqref{1} is controlled by
	\begin{align}
		\abs{\Tr\sbr{\rho\;	 \me^{\mi t H} \Pi^\perp A\Pi\me^{-\mi t H}B}}\le \norm{A}\norm{B}\abs{\Tr\sbr{\me^{-\mi t H}\rho\;	 \me^{\mi t H} \Pi^\perp }}^{1/2}\le C \norm{A}\norm{B}\sqrt{\abs{Y}}\cbr{\frac{\lam \abs{t}^d}{\nu}}^{\eta/2}.\notag
	\end{align}
	Applying Cauchy-Schwarz to the summand on line \eqref{2} yields
	\begin{align}\label{before show pi a pi is cd}
		\abs{\Tr\sbr{\rho\;	 \me^{\mi t H} \Pi A\Pi^\perp\me^{-\mi t H}B}}\le \norm{B}\abs{\Tr\sbr{\rho\;	 \me^{\mi t H} \Pi A\Pi^\perp A^\dagger\Pi \me^{-\mi t H}}}^{1/2}.
	\end{align}
    To control the r.h.s. of \eqref{before show pi a pi is cd} we show that 
    \begin{align}
         A^\dagger\Pi \me^{-\mi t H}\rho\;	 \me^{\mi t H} \Pi A
    \end{align}
    satisfies \eqref{CD rho} for $\eta$ and some new $\tilde \lambda$.
	To this end, consider $x\in\Lambda$ and $r>0$. If $B_r(x)\cap X=\emptyset$, then
\begin{align}
	 \Tr \sbr{A^\dagger \Pi\me^{-\mi t H}\rho\,\me^{\mi t H}\Pi A \,n^\eta _{B_r(x)}}&=\Tr \sbr{\Pi AA^\dagger \Pi\me^{-\mi t H}\rho\,\me^{\mi t H} \,n^\eta _{B_r(x)}}\notag\\
	&\le C\norm{A}^2 \cbr{\lam (r+vt)^d}^\eta.\notag
\end{align}
Now pick $r\ge vt$, we obtain 
\begin{align}
	\Tr \sbr{A^\dagger \Pi\me^{-\mi t H}\rho\,\me^{\mi t H}\Pi A \,n^\eta _{B_r(x)}}	&\le C\norm{A}^2 \cbr{\lam r^d}^\eta.\notag
\end{align}
If $B_r(x)\cap X\neq\emptyset$, then
\begin{align}
	\Tr \sbr{A^\dagger \Pi\me^{-\mi t H}\rho\,\me^{\mi t H}\Pi A \,n^\eta _{B_r(x)}}&\le \Tr \sbr{A^\dagger \Pi\me^{-\mi t H}\rho\,\me^{\mi t H}\Pi A \,n^\eta _{x[r+d(X)]}}\notag\\
	&=\Tr \sbr{\Pi AA^\dagger \Pi\me^{-\mi t H}\rho\,\me^{\mi t H} \,n^\eta _{x[r+d(X)]}}\notag\\
   &\le C\norm{A}^2 \cbr{\lam (r+vt+d(X))^d}^\eta.\notag
\end{align}
Pick $r\ge vt+1$, since by definition $d(Z)\ge1$ for every $X\subset \Lambda$,
\begin{align}
\cbr{\lam (r+vt+d(X))^d}^\eta\le C\cbr{\lam \,d(X)^{ d}r^{ d}}^\eta\notag
\end{align}
Thus, for all $r\ge vt+1$, it holds,
\begin{align}\label{rho A is CD}
	\Tr \sbr{A^\dagger \Pi\me^{-\mi t H}\rho\,\me^{\mi t H}\Pi A \,n^\eta _{B_r(x)}}\le C \norm{A}^2\cbr{\lam \,d(X)^{ d}r^{ d}}^\eta.
\end{align}
	Inequality \eqref{rho A is CD} allows us to control the r.h.s. of \eqref{before show pi a pi is cd} as follows
	\begin{align}
		\norm{B}\abs{\Tr\sbr{\rho\;	 \me^{\mi t H} \Pi A\Pi^\perp A^\dagger\Pi \me^{-\mi t H}}}^{1/2}\le C \norm{A}\norm{B}\abs{Y}^{1/2}d(X)^{d\eta /2}\cbr{\frac{\lam \cbr{\abs{t}+1}^{d}}{\nu}}^{\eta/2}.\notag
	\end{align}
	This concludes the proof.
\end{proof}

In the following proposition we approximate $\tau_t (\bar A)$ by $\bar \tau_t (\bar A)$ for every $A\in\Ainv{X}$, making use of the particle propagation bounds we developed in \secref{sec PPB}. This completes steps (2) and (4) in \eqref{strategy LRB}.

\begin{proposition}\label{prop H to bar H}
In the same setting as in \lemref{lemma bar A}, there exists a positive constant ${C=C(d,J,\eta)}$ such that  the following holds
	\begin{align}
		&\abs{\Tr\sbr{\rho\cbr{	 \tau_t (\bar A)-	\bar \tau_t (\bar A)}B}}\le C  \norm{A}\norm{B}\abs{Y}d(X)^{d\eta /2}\cbr{\frac{\lam}{\nu}}^{\eta /2}(\abs{t}+1)^{d\eta/2}\cbr{\nu \abs{t}+1}\notag
	\end{align}
    for all $t\in \R$
\end{proposition}

 We aim to approximate $\bar \tau_t (\bar A)$ by $\bar \tau_t^R(\bar A)$  for every ${A\in\Ainv{X}}$. 
A key ingredient to achieve this result are Lieb--Robinson bounds for short-range and bounded interactions. 
For the convenience of the reader we include the statement.
\begin{theorem}[Lieb--Robinson bounds for short-range bounded interactions \cite{nachtergaele2019quasi}]\label{teo LRB bounded}
    Consider $\cP_0(\Lambda)$ the collection of all sets in $\Lambda$ and $I\subset \R$ an interval. 
    Define the algebra of local observables as
    \begin{align}
        \cA^{\mathrm{loc}}:=\bigcup _{Z\in \cP_0(\Lambda)}\cA_Z.\notag
    \end{align}
    Fix a possibly time-dependent Hamiltonian,
    \begin{align}
        H(t)=\sum_{\substack{Z\in \cP_0(\Lambda) }}\Phi(Z,t),\notag
    \end{align}
    where $\Phi: \cP_0(\Lambda)\times I \to \cA^{\mathrm{loc}}$ such that 
    \begin{enumerate}
        \item $\Phi(Z,t)^\dagger=\Phi(Z,t)\in \cA_Z$ for all $Z\in\cP_0(\Lam)$ and $t\in I$.
        \item For every $Z\in\cP_0(\Lam)$, $\Phi(Z, \dot): I\to \cA_Z$ is strongly continuous.
        \item The norm 
        \begin{align}
            \norm{\Phi(t)}_F:=\sup_{x,y\in\Lam}\frac{1}{F(d(x,y)}\sum_{\substack{Z\in \cP_0(\Lambda)\\x,y\in Z}}\norm{\Phi(Z,t)}\notag
        \end{align}
        is bounded, for every $t\in I$, where $F:[0,\infty)\to[0,\infty) $ is defined as
        \begin{align}
            F(r):=\me^{-r}(1+r)^{-2d}.\notag
        \end{align}
         \end{enumerate} 
        Then, there exist positive constants $C, v_{\mathrm{LR}}$ such that for every $X,Y\in \cP_0(\Lambda) $, with $d(X,Y)>0$, and any $A\in \cA_X$, $B\in\cA_Y$, the following holds
        \begin{align}
            \norm{\sbr{\tau_{t,s}(A),B}}\le C \norm{A}\norm{B}\min\Set{\abs{X},\abs{Y}}\me^{v_{\mathrm{LR}}\abs{t-s}-d(X,Y)},\notag
        \end{align}
   for all $t,s\in I$. Here $\tau_{t,s}$ indicates the time evolution induced by $H$. 
\end{theorem}

Consider now $\Pi_{Y,\nu}$ with $Y=X[R+3]$ for appropriate choice of $R$.
Thanks to \thmref{teo LRB bounded} we derive the following proposition and complete step (3) in \eqref{strategy LRB}.  
\begin{proposition}\label{prop bar H to H R}
    There exist positive constants $C,v$ depending on $J,d$, such that, for every $X\subset\Lambda$, $\nu>0$,  $R>2$, and $A\in\Ainv{X}$, the following holds
	\begin{align}
		\norm{\bar \tau_t (\bar A)-\bar \tau_t^R(\bar A)}\le  C\abs{X}\norm{A}R^d\me^{-R}\cbr{\me^{ J\nu vt}-1},\notag
	\end{align}
	for every $ \abs{t}\le (R-2)/\nu v J$. 
\end{proposition}
The idea of the proof of \propref{prop bar H to H R} is to first go in the interaction picture where the potential plays the role of the unperturbed part. Then, after applying Duhamel's formula, we bound the commutator between A time evolved through the interaction picture Hamiltonian restricted on $X[R]$ and the interactions of the interaction picture Hamiltonian that are supported on sets intersecting $X[R]$ and its complement. To control such a commutator we invoke \thmref{teo LRB bounded}. Notice that we can apply such a result since the boson truncation due to $\Pi$ implies that the interaction picture Hamiltonian restricted on $X[R]$ has bounded interaction norm.

\subsection{Proof of \thmref{teo LRB}}\label{sec proof LRB}

\begin{proof}
	We implement the step we laid out in \eqref{strategy LRB}.
	By triangular inequality it follows
	\begin{align}
		\abs{\Tr\sbr{\rho\cbr{	 \tau_t (A)-	 \tau^R_t (A)}B}}&\le \abs{\Tr\sbr{\rho\cbr{	 \tau_t (A)-	 \tau_t (\bar A)}B}}\label{1 lrb}\\
		&+\abs{\Tr\sbr{\rho\cbr{	 \tau_t (\bar A)-	\bar \tau_t (\bar A)}B}}\label{2 lrb}\\
		&+\abs{\Tr\sbr{\rho\cbr{	 \bar \tau_t (\bar A)-	\bar \tau^R_t (\bar A)}B}}\label{3 lrb}\\
		&+\abs{\Tr\sbr{\rho\cbr{		\bar \tau^R_t (\bar A)-  \tau^R_t (\bar A)}B}}\label{4 lrb}\\
		&+\abs{\Tr\sbr{\rho\cbr{		\tau^R_t (\bar A)-  \tau^R_t ( A)}B}}\label{5 lrb}.
	\end{align}
	We control lines \eqref{1 lrb} and \eqref{5 lrb} thanks to  \lemref{lemma bar A}.
	Set $\nu=\frac{R-2}{2vJ\max(|t|,1)}$ and $Y=X[R+3]$ and apply \propref{prop H to bar H} to control \eqref{2 lrb} and \eqref{4 lrb}. We apply \propref{prop bar H to H R} to bound \eqref{3 lrb}.
	To conclude the proof we recall 
	\begin{align}
		\norm{A}_1:=\Tr\sbr{\sqrt{A^\dagger A}}= \sup \Set{\abs{\Tr\sbr{A B}}\,:\, B\in\cB, \norm{B}=1}.\notag
	\end{align}

\end{proof}

\subsection{Proof of \propref{prop H to bar H}}\label{sec Proof of prop H to bar H}

\begin{proof}
    Without loss of generality we consider $t\ge 0$.
	Throughout this proof we will write $\Pi\equiv\Pi_{Y,\nu}$, and we will consider $t\ge 0$.
	By triangular inequality it holds
	\begin{align}
		&\abs{\Tr\sbr{\rho\cbr{	 \tau_t (\bar A)B-	\bar \tau_t (\bar A)B}}}\notag\\
        &\qquad\le \underbrace{\abs{\Tr \sbr{\rho\cbr{\me^{\mi t\bar H}\bar A\me^{-\mi t\bar H}-\me^{\mi t H}\bar A\me^{-\mi t\bar H}}B}}}_{\mathrm{I}}	 +\,\underbrace{\abs{\Tr \sbr{\rho\cbr{\me^{\mi t H}\bar A\me^{-\mi t H}-\me^{\mi t H}\bar A\me^{-\mi t\bar H}}B}}}_{\mathrm{II}}\label{II lr}.
	\end{align}
We start bounding term $(\mathrm{I})$. 
Again, by triangular inequality	it holds
\begin{align}
	\mathrm{(I)}&\le \underbrace{\abs{\Tr \sbr{\rho\cbr{\me^{\mi t\bar H}-\me^{\mi t H}} \Pi \bar A\me^{-\mi t\bar H} B}}}_{\mathrm{Ia}}+\,\underbrace{\abs{\Tr \sbr{\rho\,\me^{\mi t\bar H}\Pi^\perp \bar A\me^{-\mi t\bar H} B}}}_{\mathrm{Ib}}+\,\underbrace{\abs{\Tr \sbr{\rho\,\me^{\mi t H}\Pi^\perp \bar A\me^{-\mi t\bar H}B} }}_{\mathrm{Ic}}\notag.
\end{align}
Applying consecutively the properties of $\Pi$ \eqref{prop p perp}, Cauchy-Schwarz,  Markov's inequality, and the assumption of controlled density \eqref{CD rho}, we can control term $(\mathrm{Ib})$ as 
\begin{align}\label{Ib lrb final}
	\mathrm{(Ib)}=\abs{\Tr \sbr{\rho\,\Pi^\perp \bar A\me^{-\mi t\bar H} B}}&\le \norm{A}\norm{B}\sqrt{\Tr\sbr{\rho\, \Pi^\perp}}\le \norm{A}\norm{B}\sqrt{ \abs{Y}}\cbr{\frac{\lam }{\nu}}^{\eta/2}.
\end{align}
The reasoning above, coupled with \thmref{theo final particle} for $\beta =d\eta+1 $ and $v=2\kappa$, yields the following bound on term $(\mathrm{Ic})$
\begin{align}\label{Ic lrb final}
	\mathrm{(Ic)}\le \norm{A}\norm{B}\sqrt{\Tr\sbr{\rho\, \me^{\mi t H}\Pi^\perp \me^{-\mi t H}}}\le C \norm{A}\norm{B}\sqrt{ \abs{Y}}\cbr{\frac{\lam t^d}{\nu}}^{\eta/2}.
\end{align}
To control term $(\mathrm{Ia})$ we employ Duhamel's formula
\begin{align}\label{Ia 1}
	\mathrm{(Ia)}&\le \int_0^t \abs{\Tr \sbr{\rho\,\me^{\mi(t-t') H} \cbr{H\Pi-\bar H}\me^{-\mi t'\bar H} \bar A\me^{-\mi t\bar H} B}} \md t'.
\end{align}
Notice that  $\sbr{V_{xy},n_x}=0$ implies $\Pi V\Pi^\perp=0$. 
Furthermore, there hold $\sbr{\Pi_{Y\setminus \Set{x,y},\nu}, T_{xy}}=0$ and $\Pi^\perp \Pi_{Y\setminus \Set{x,y},\nu}=\Pi^\perp_{\Set{x,y},\nu}\Pi_{Y\setminus \Set{x,y},\nu}$. 
Thus,
\begin{align}\label{diff H HP}
	H\Pi-\bar H = \Pi^\perp H \Pi= \sum_{\substack{x,y\in Y[1]\\x\sim y}}\Pi_{\Set{x,y},\lambda}^\perp T_{xy}\Pi .
\end{align}
Thanks to \eqref{diff H HP} and Cauchy-Schwarz, it follows from \eqref{Ia 1} that
\begin{align}\label{Ia 2}
	\mathrm{(Ia)}&\le \int_0^t \sum_{\substack{x,y\in Y[1]\\x\sim y}} \abs{\Tr \sbr{\rho\,\me^{\mi(t-t') H} \Pi_{\Set{x,y},\lambda}^\perp T_{xy}\Pi\me^{-\mi t'\bar H} \bar A\me^{-\mi t\bar H} B}}\md t'\notag\\
	&\le \norm{A}\norm{B}\int_0^t \sum_{\substack{x,y\in Y[1]\\x\sim y}}\norm{\Pi T_{xy}} \cbr{\Tr \sbr{\rho\,\me^{\mi(t-t') H} \Pi_{\Set{x,y},\lambda}^\perp \me^{\mi(t-t') H}}}^{1/2}\md t'.
\end{align}
By definition of $T_{xy}$ and the properties of the bosonic operators, it holds
\begin{align}\label{Ia piece 1}
	\norm{\Pi T_{xy}} \le 2\sqrt{2}\,\nu.
\end{align}
To bound the second integrated factor in \eqref{Ia 2} we make use again of \thmref{theo final particle}
\begin{align}\label{Ia piece 2}
	\Tr \sbr{\rho\,\me^{\mi(t-t') H} \Pi_{\Set{x,y},\lambda}^\perp \me^{\mi(t-t') H}}\le 2C \cbr{\frac{\lam (t-t')^d}{\nu}}^\eta.
\end{align}
Applying \eqref{Ia piece 1} and \eqref{Ia piece 2} to \eqref{Ia 2} we achieve
\begin{align}\label{Ia lrb final}
	\mathrm{(Ia)}&\le C \lam^{\eta/2} \norm{A}\norm{B}\abs{Y} \nu t\cbr{\frac{t^{d}}{\nu}}^{\eta/2}.
\end{align}
Above we employed $\abs{\Set{x,y\in Y[1]\,:\,x\sim y}}\le C_d \abs{Y}$. 
Inequalities \eqref{Ib lrb final}, \eqref{Ic lrb final}, and \eqref{Ia lrb final} imply
\begin{align}\label{I lrb final}
	(\mathrm{I})\le C \lam^{\eta/2}\norm{A}\norm{B}\abs{Y}\cbr{\cbr{\nu t+1}\cbr{\frac{t^{d}}{\nu}}^{\eta/2}+\nu^{-\eta/2}}.
\end{align}
To conclude the proof we need to control \eqref{II lr}. Thanks to triangular inequality it holds
\begin{align}
	(\mathrm{II})&\le \underbrace{\abs{\Tr \sbr{\rho\cbr{\me^{\mi t H}\bar A\cbr{\me^{-\mi t H}-\me^{-\mi t\bar H}}}\Pi B}}}_{\mathrm{IIa}} +\,\underbrace{\abs{\Tr \sbr{\rho\,\me^{\mi t H}\bar A\me^{-\mi t\bar H}\Pi^\perp B}}}_{\mathrm{IIb}}\,+ \underbrace{\abs{\Tr \sbr{\rho\,\me^{\mi t H}\bar A\me^{-\mi t H}\Pi^\perp B}}}_{\mathrm{IIc}}\notag
\end{align}
We start bounding term $(\mathrm{IIc})$ applying Cauchy-Schwarz,
\begin{align}\label{before prop Ae}
	(\mathrm{IIc})\le \norm{B}\sqrt{\abs{\Tr \sbr{\bar A^\dagger\me^{-\mi t H}\rho\,\me^{\mi t H}\bar A\me^{-\mi t H}\Pi^\perp \me^{\mi t H}}}}.
\end{align}

Inequality \eqref{rho A is CD} implies $\bar A^\dagger\me^{-\mi t H}\rho\,\me^{\mi t H}\bar A$ satisfies \eqref{CD rho} for $\eta$ and some $\tilde \lambda:=C\norm{A}^{2/\eta} \lambda d(X)^d$.
This, together with \thmref{theo final particle} with $v=2\kappa$, $\delta=1/2$,  $r= vt+1$, and $R=2r$, implies 
\begin{align}
	\Tr \sbr{\me^{\mi t H}\bar A^\dagger\me^{-\mi t H}\rho\,\me^{\mi t H}\bar A \me^{-\mi t H}\Pi^\perp}\le C\abs{Y}\norm{A}^2\cbr{\lam \,d(X)^{ d}}^\eta\cbr{\frac{(t+1)^d}{\nu}}^\eta\label{tra a perp}.
\end{align}
Thus, applying \eqref{tra a perp} to \eqref{before prop Ae} leads to
\begin{align}\label{IIc lr final}
	(\mathrm{IIc})\le C \norm{B}\norm{A}\abs{Y}^{1/2}\,d(X)^{\eta d/2}\cbr{\frac{\lam (t+1)^{d}}{\nu}}^{\eta/2}.
\end{align}
Recalling \eqref{prop p perp}, applying \eqref{rho A is CD} together with Markov's inequality, we derive the following bound on term $(\mathrm{IIb})$
\begin{align}\label{IIb lr final}
	(\mathrm{IIb})\le \norm{B}\sqrt{\abs{\Tr \sbr{\bar A^\dagger\me^{-\mi t H}\rho\,\me^{\mi t H}\bar A\Pi^\perp }}}\le  C \norm{B}\norm{A}\abs{Y}^{1/2}\,d(X)^{\eta d/2}\cbr{\frac{\lam t^{d}}{\nu}}^{\eta/2}.
\end{align}
To control term $(\mathrm{IIa})$ we employ again the Duhamel's formula, Cauchy-Schwarz, \eqref{diff H HP}, \eqref{Ia piece 1}, and the strategy above, to obtain
\begin{align}\label{IIa lr final}
	(\mathrm{IIa})&\le 2\sqrt{2} \,\nu \norm{B}\int_0^t \sum_{\substack{x,y\in Y[1]\\x\sim y}}\cbr{\tr\sbr{\rho\, \me^{\mi t H}A\Pi^\perp_{\Set{x,y},\nu}A^\dagger \me^{-\mi t H}}}^{1/2}\md t'\notag\\
	&\le  C \norm{A}\norm{B}\abs{Y}{ \,d(X)^{ d\eta/2}}\,\nu t\cbr{\frac{\lam t^d}{\nu}}^{\eta/2}.
\end{align}
Combining \eqref{IIc lr final}--\eqref{IIa lr final} yields
\begin{align}\label{II lr final}
	(\mathrm{II})\leq C \norm{A}\norm{B}\abs{Y}{ \,d(X)^{ d\eta/2}}(\nu t+1)\cbr{\frac{\lam (t+1)^{d}}{\nu}}^{\eta/2} .
\end{align}
Inequalities \eqref{I lrb final} and \eqref{II lr final} yield the desired inequality.

\end{proof}

\subsection{Proof of \propref{prop bar H to H R}}\label{sec Proof of prop bar H to H R}

\begin{proof}
    Without loss of generality we assume $t\ge 0$ and that $R$ is an integer.
	We start by going into the interaction picture where $\bar T$ plays the role of perturbation and $\bar V$ of the unperturbed part.
	The interaction picture propagator is defined as 
	\begin{align}
		U_{0,t}:=\me^{\mi t\bar H}\me^{-\mi t\bar V}=1-\mi \int_0^t\me^{\mi t\bar V}\bar T\me^{-\mi t'\bar V}U_{0,t'}\, \md t',\notag
	\end{align}
	where the last equality is a consequence of the Duhamel's formula.
The time-dependent generator of $U_{0,t}$ is given by the following interaction picture Hamiltonian
\begin{align}
	\bar T^\ip(t):= \me ^{\mi t\bar V}\bar T\, \me ^{-\mi t\bar V}=\sum_{x\sim y} \me ^{\mi t\bar V}\bar T_{xy} \me ^{-\mi t\bar V}=\sum_{x\sim y}\bar T^\ip_{xy}(t),\notag
\end{align}
where we defined
\begin{align}
	\bar T^\ip_{xy}(t):=\me ^{\mi t\bar V_{\cB_{x,y}}}\bar T_{xy} \me ^{-\mi t\bar V_{\cB_{x,y}}},\qquad \cB_{x,y}:=B_1(x)\cup B_1(y),\notag
\end{align}
and $V_{\cB_{x,y}}$ is to be understood as in \eqref{T V on X}.
Above we used the fact that the potential $V$ is a sum of commuting local terms and that for operators $A,B$ supported on disjoint sets,  $\sbr{\bar A, \bar B}=0$.
Given an operator $A\in \Ainv{X}$, its time evolution under the full Hamiltonian $\bar H$ can then be written as 
\begin{align}\label{int pic H}
	\me^{\mi t\bar H}\bar A\me^{-\mi t\bar H}&=U_{0,t}\me^{\mi t\bar V}\bar A\me^{-\mi t\bar V}U_{t,0}=U_{0,t}\me^{\mi t\bar V_{X[1]}}\bar A\me^{-\mi t\bar V_{X[1]}}U_{t,0}.
\end{align}
We write
\begin{align}\label{A V}
	\me^{\mi t\bar V_{X[1]}}\bar A\me^{-\mi t\bar V_{X[1]}}=\overline{\me^{\mi t\bar V_{X[1]}} A\me^{-\mi t\bar V_{X[1]}}} =: \overline{A_V}	.
\end{align}
To conclude the proof we need to approximate $U_{0,t}\,\overline{A_V}\, U_{t,0}$ by $U^{R,X}_{0,t}\,\overline{A_V} \,U^{R,X}_{t,0}$, where $U^{R,X}_{0,t}$ is the dynamics generated by 
\begin{align}
	\bar T_{R,X}^\ip(t):=\sum_{\substack{x\sim y\\x,y\in  X[R]}} \bar T^\ip_{xy}(t)=\sum_{\substack{x\sim y\\\cB_{x,y}\subset X[R+1]}} \bar T^\ip_{xy}(t) .\notag
\end{align}
Analogously we define 
\begin{align}\label{int pic ham R}
	\bar T_{R,X^c}^\ip(t):=\sum_{\substack{x\sim y\\\cB_{x,y}\subset X[R+1]^c}} \bar T^\ip_{xy}(t)
\end{align}
 and $U^{R,X^c}_{0,t}$ is the dynamics $\bar T_{R,X^c}^\ip(t)$ generates.
Notice that, since $R>1$,
\begin{align}
	\sbr{\,\overline{A_V}, \bar T_{R,X^c}^\ip(t)}=0.\notag
\end{align}
In fact,
\begin{align}
	\overline{A_V}\, \bar T_{R,X^c}^\ip(t)&=\sum_{\substack{x\sim y\\\cB_{x,y}\subset X[R+1]^c}}\me^{\mi t\bar V_{X[1]}}\bar A\me^{-\mi t\bar V_{X[1]}}\me ^{\mi t\bar V_{\cB_{x,y}}}\bar T_{xy} \me ^{-\mi t\bar V_{\cB_{x,y}}}\notag\\
	&=\sum_{\substack{x\sim y\\\cB_{x,y}\subset X[R+1]^c}}\me ^{\mi t\bar V_{\cB_{x,y}}}\bar T_{xy} \me ^{-\mi t\bar V_{\cB_{x,y}}}\me^{\mi t\bar V_{X[1]}}\bar A\me^{-\mi t\bar V_{X[1]}}.\notag
\end{align}
Thus, we can write
\begin{align}
	U^{R,X}_{0,t}\,\overline{A_V} \,U^{R,X}_{t,0}=U^{R,X}_{0,t}U^{R,X^c}_{0,t}\,\overline{A_V}\, U^{R,X^c}_{t,0}U^{R,X}_{t,0}.\notag
\end{align}
This fact, together with the Duhamel's formula, yields 
\begin{align}
	U_{0,t}&\overline{A_V} \,U_{t,0}-U^{R,X}_{0,t}\overline{A_V} \,U^{R,X}_{t,0}\notag\\
	&=-\mi  \sum_{\substack{x\sim y\\\cB_{x,y}\cap X[R+1]\neq\emptyset\\\cB_{x,y}\cap X[R+1]^c\neq\emptyset}} \int_0^t U_{t',t}U^{R,X}_{0,t'}\sbr{U^{R,X}_{t',0}\;\bar T^\ip_{xy}(t')\;U^{R,X}_{0,t'},A}U^{R,X}_{t',0}U_{t,t'}\md t'.\notag
\end{align}
Therefore,
\begin{align}\label{after op norm}
	\norm{U_{0,t}\overline{A_V} \,U_{t,0}-U^{R,X}_{0,t}\overline{A_V} \,U^{R,X}_{t,0}}\le \sum_{\substack{x\sim y\\\cB_{x,y}\cap X[R+1]\neq\emptyset\\\cB_{x,y}\cap X[R+1]^c\neq\emptyset}}\int_0^t\norm{\sbr{U^{R,X}_{t',0}\;\bar T^\ip_{xy}(t')\;U^{R,X}_{0,t'},\overline{A_V}}}\md t'.
\end{align}
Let us focus on the norm appearing on the r.h.s. of \eqref{after op norm}. Since all the operators involved have support contained in $X[R+3]$, the commutator acts trivially on $X[R+3]^c$. 
Thus,
\begin{align}
	\norm{\sbr{U^{R,X}_{t',0}\;\bar T^\ip_{xy}(t')\;U^{R,X}_{0,t'},\overline{A_V}}}&=\sup_{\substack{\psi\in \cF_\Lambda\\\norm{\psi}=1}}
	\abs{\br{\psi,\sbr{U^{R,X}_{t',0}\;\bar T^\ip_{xy}(t')\;U^{R,X}_{0,t'},\overline{A_V}}\psi }}\notag\\
	&=\sup_{\substack{\psi\in \cF_{X[R+3]}\\\norm{\psi}=1}
	}\abs{\br{\psi,\sbr{U^{R,X}_{t',0}\;\bar T^\ip_{xy}(t')\;U^{R,X}_{0,t'},\overline{A_V}}\psi }},\notag
\end{align}
where $\mathcal F_{X[R+3]}=  \C\oplus\bigoplus_{N=1}^\infty \ell_s^2(X[R+3]^N)$.
Since every operator appearing above is conjugated by the projector $\Pi_{X[R+3],\nu}$, then
\begin{align}
	\sup_{\substack{\psi\in \cF_{X[R+3]}\\\norm{\psi}=1}} \abs{\br{\psi,\sbr{U^{R,X}_{t',0}\;\bar T^\ip_{xy}(t')\;U^{R,X}_{0,t'},\overline{A_V}}\psi }}&=\sup_{\substack{\psi\in \cF_{X[R+3]}\\ \norm{\psi}=1\\ \br{\psi,n_x \psi}\le \nu, \, \forall x\in X[R+3]}}\abs{\br{\psi,\sbr{U^{R,X}_{t',0}\;\bar T^\ip_{xy}(t')\;U^{R,X}_{0,t'},\overline{A_V}}\psi }}\label{before lrb}.
\end{align}
Notice that, for a given $\psi$ chosen as above and $A\in \Ainv{X}$, it holds
\begin{align}\label{pi dont change supp}
	\bar A \psi= \Pi_{X,\nu}A\Pi_{X,\nu}\psi.
\end{align}
Then, for every $\psi$ as above, the conjugation by the projector does not change the support of the operator.
Thus, $\bar T_{xy}^\ip(t)$ is effectively supported on $\cB_{x,y}$ and
\begin{align}
	\norm{\bar T_{xy}^\ip(t)}\le \norm{\bar T_{xy}}\le 4\abs{J}\nu,\quad \forall t\in \R.\notag
\end{align}
Therefore we can bound the interaction norm of the interaction picture Hamiltonian as follows
\begin{align}\label{finite rangeness of int T}
	\max_{x',y'\in\Lam}\sum_{\substack{x\sim y\\ x',y'\in \cB_{x,y}}}\me^{d(x',y')}(1+d(x',y'))^{2d}\norm{\bar T_{xy}^\ip(t)}\le C_d\abs{J}\nu .
\end{align}
Due to the summation constraint in \eqref{after op norm} and the observation \eqref{pi dont change supp} , $\bar T^\ip_{xy}(t')$ and $\overline{A_V}$ are supported on sets distant at least $R-2$ from each other. 
This and \eqref{finite rangeness of int T} allow us to apply the Lieb--Robinson bounds in \thmref{teo LRB bounded} to control the r.h.s. of \eqref{before lrb}.
Namely, there exist constants $v_{\mathrm{LR}},C$ such that, for $t'\le (R-2)/v_{\mathrm{LR}}\nu J$,
\begin{align}
	\norm{\sbr{U^{R,X}_{t',0}\;\bar T^\ip_{xy}(t')\;U^{R,X}_{0,t'},\overline{A_V}}}&\le C\abs{B_{x,y}}\norm{A}\norm{\bar T^\ip_{xy}(t')}\me^{J\nu v_{\mathrm{LR}}t'-R+3}\notag\\
	&\le C_1 J\nu\norm{A}\me^{J\nu v_{\mathrm{LR}}t'-R}.\notag
\end{align}
Thus,
\begin{align}
	\norm{U_{0,t}\overline{A_V} \,U_{t,0}-U^{R,X}_{0,t}\overline{A_V} \,U^{R,X}_{t,0}}&\le C_2J\nu\norm{A} \sum_{\substack{x\sim y\\\cB_{x,y}\cap X[R+1]\neq\emptyset\\\cB_{x,y}\cap X[R+1]^c\neq\emptyset}}\int_0^t\me^{ J\nu v_{\mathrm{LR}}t'-R}\md t'\notag\\
	&\le C_3\abs{X[R+3]}\norm{A}\cbr{\me^{ J\nu v_{\mathrm{LR}}t}-1}\me^{-R}\notag\\
	&\le C_4\abs{X}\norm{A}R^d\cbr{\me^{ J\nu v_{\mathrm{LR}}t}-1}\me^{-R},\notag
\end{align}
for all $J\nu v  t\le R-2$.
Now recall \eqref{int pic H} and \eqref{A V} to write $U_{0,t}\overline{A_V} \,U_{t,0}=\bar\tau_t (\bar A)$. Furthermore, since $U^{R,X}$ is the dynamics generated by \eqref{int pic ham R}, it holds that $U^{R,X}_{0,t}\overline{A_V} \,U^{R,X}_{t,0}=\bar\tau_t^R (\bar A)$. This concludes the proof of \propref{prop bar H to H R}.
\end{proof}

\subsection{Proof of \lemref{lemma comm exp}}

\begin{proof}
	By definition \eqref{dGdef},  $\mathrm{d}\Gamma(g)$ commutes with any $n_z$ and therefore with the potential part of the Hamiltonian. Thus, 
	\begin{align} \label{H0Rf-com}
		[H_\Lam,\mathrm{d}\Gamma(g)] = J\sum_{x\in\Lambda,y\in\Lambda}  [a_x^*a_y, \dG(g) ]&=J\sum_{x,y\in\Lambda}\sum_{z\in\Lambda}  g(z)[a_x^*a_y, a_z^*a_z].\end{align}
	By the canonical commutation relation, it holds
	\begin{align}
		[a_x^*a_y, a_z^*a_z] =\begin{cases} -a_x^*a_y, & \mbox{if } z=x \\ a_x^*a_y, & \mbox{if } z=y\\ 0,& \mbox{otherwise} \end{cases} \notag
	\end{align}
After relabeling the sum in \eqref{H0Rf-com}, the relations above yield the desired equality \eqref{HRf-com}.
\end{proof}

\section*{Acknowledgments}  
		The research of ML and CR is supported by the DFG through the grant TRR 352 – Project-ID 470903074. The research of ML is further supported by the DFG through the grant FOR 5413 – Project-ID 465199066 and by the European Union through ERC Starting Grant MathQuantProp, Grant Agreement 101163620.\footnote{Views and opinions expressed are however those of the authors only and do not necessarily
			reflect those of the European Union or the European Research Council Executive Agency. Neither
			the European Union nor the granting authority can be held responsible for them.}


\end{document}